\documentclass[11pt,draftclsnofoot,peerreviewca,letterpaper,onecolumn]{IEEEtran}

\usepackage{cite}

\usepackage{graphicx,color,epsfig,rotating,subfigure}
\usepackage{amsfonts,amsmath,amssymb}
\usepackage{algorithm,algorithmic}
\usepackage{subfigure}

\long\def\comment#1{}

\newfont{\bbb}{msbm10 scaled 700}

\newfont{\bb}{msbm10 scaled 1100}
\newcommand{\CC}{\mbox{\bb C}}
\newcommand{\PP}{\mbox{\bb P}}
\newcommand{\RR}{\mbox{\bb R}}

\newcommand{\ZZ}{\mbox{\bb Z}}
\newcommand{\FF}{\mbox{\bb F}}

\newcommand{\EE}{\mbox{\bb E}}

\newcommand{\av}{{\bf a}}
\newcommand{\bv}{{\bf b}}
\newcommand{\cv}{{\bf c}}
\newcommand{\dv}{{\bf d}}

\newcommand{\hv}{{\bf h}}

\newcommand{\qv}{{\bf q}}
\newcommand{\rv}{{\bf r}}

\newcommand{\tv}{{\bf t}}
\newcommand{\uv}{{\bf u}}
\newcommand{\wv}{{\bf w}}
\newcommand{\vv}{{\bf v}}
\newcommand{\xv}{{\bf x}}
\newcommand{\yv}{{\bf y}}
\newcommand{\zv}{{\bf z}}
\newcommand{\zerov}{{\bf 0}}

\newcommand{\Am}{{\bf A}}
\newcommand{\Bm}{{\bf B}}

\newcommand{\Dm}{{\bf D}}

\newcommand{\Fm}{{\bf F}}
\newcommand{\Gm}{{\bf G}}
\newcommand{\Hm}{{\bf H}}
\newcommand{\Id}{{\bf I}}

\newcommand{\Lm}{{\bf L}}

\newcommand{\Qm}{{\bf Q}}

\newcommand{\Tm}{{\bf T}}
\newcommand{\Um}{{\bf U}}

\newcommand{\Xm}{{\bf X}}
\newcommand{\Ym}{{\bf Y}}
\newcommand{\Zm}{{\bf Z}}

\newcommand{\Ac}{{\cal A}}

\newcommand{\Cc}{{\cal C}}

\newcommand{\Ic}{{\cal I}}

\newcommand{\Lc}{{\cal L}}
\newcommand{\Mc}{{\cal M}}
\newcommand{\Nc}{{\cal N}}

\newcommand{\Sc}{{\cal S}}

\newcommand{\Uc}{{\cal U}}

\newcommand{\Vc}{{\cal V}}

\newcommand{\lambdav}{\hbox{\boldmath$\lambda$}}

\newcommand{\nuv}{\hbox{\boldmath$\nu$}}
\newcommand{\muv}{\hbox{\boldmath$\mu$}}
\newcommand{\zetav}{\hbox{\boldmath$\zeta$}}

\newcommand{\diag}{{\hbox{diag}}}
\renewcommand{\det}{{\hbox{det}}}
\newcommand{\trace}{{\hbox{tr}}}

\newcommand{\SNR}{{\sf SNR}}

\renewcommand{\Re}{{\rm Re}}

\newcommand{\eqdef}{\stackrel{\Delta}{=}}

\newcommand{\herm}{{\sf H}}

\newcommand{\transp}{{\sf T}}

\newtheorem{theorem}{Theorem}
\newtheorem{lemma}{Lemma}

\newtheorem{proof}{Proof}
\newtheorem{remark}{Remark}

\newcommand{\argmin}{\operatornamewithlimits{argmin}}

\begin{document}

\title{Compute-and-Forward Strategies for Cooperative Distributed Antenna Systems}

\author{\authorblockN{Song-Nam~Hong,~\IEEEmembership{Student Member,~IEEE,}
        and~Giuseppe~Caire,~\IEEEmembership{Fellow,~IEEE}}
\authorblockA{Department of Electrical Engineering, University of Southern California, Los Angeles, CA, USA}
\authorblockA{(e-mail: \{songnamh, caire\}$@$usc.edu)}
\thanks{This research was supported in part by the KCC (Korea Communications Commission), Korea, under the R\&D program supervised by the KCA (Korea Communications Agency) (KCA-2011-11921-04001).}}

\maketitle

\newpage

\begin{abstract}
We study a distributed antenna system where $L$ antenna terminals (ATs) are connected to a Central Processor (CP) via digital error-free links
of finite capacity $R_{0}$, and serve $K$ user terminals (UTs).
This model has been widely investigated  both for the uplink (UTs to CP) and for the downlink (CP to UTs),
which are instances of the general multiple-access relay and broadcast relay networks.
We contribute to the subject in the following ways: 1) for the uplink, we apply the ``Compute and Forward'' (CoF) approach and
examine the corresponding system optimization at finite SNR; 2) For the downlink, we propose a novel precoding scheme nicknamed
``Reverse Compute and Forward" (RCoF); 3) In both cases, we present low-complexity versions of CoF and RCoF based on standard
scalar quantization at the receivers, that lead to discrete-input discrete-output symmetric memoryless channel models
for which near-optimal performance can be achieved by standard single-user linear coding; 4) For the case of large $R_0$, we propose a novel
``Integer Forcing Beamforming'' (IFB) scheme that generalizes the popular zero-forcing beamforming and achieves sum rate performance close to
the optimal Gaussian Dirty-Paper Coding.

The proposed uplink and downlink system optimization focuses specifically on the ATs and UTs selection problem. In both cases,
for a given set of transmitters, the goal consists of selecting a subset of the receivers such that the corresponding system matrix has full rank and
the sum rate is maximized.  We present low-complexity ATs and UTs selection schemes and demonstrate, through Monte Carlo
simulation in a realistic environment with  fading and shadowing, that the proposed schemes essentially eliminate the problem
of rank deficiency of the system matrix and greatly mitigate the non-integer
penalty affecting CoF/RCoF at high SNR. Comparison with other state-of-the art information theoretic schemes,
such as ``Quantize reMap and Forward'' for the uplink and ``Compressed Dirty Paper Coding'' for the downlink, show competitive
performance of the proposed approaches with significantly lower complexity.
\end{abstract}


\newpage

\section{Introduction}\label{sec:Intro}

\IEEEPARstart{A } cloud base station is a {\em Distributed Antenna System} (DAS)
formed by a number of simple antenna terminals (ATs) \cite{bell}, spatially distributed over a certain area, and connected to a central processor (CP)
via wired backhaul \cite{Flanagan,Lin,Marict}. Cloud base station architectures differ by the type of processing made at the
ATs and at the CP, and by the type of wired backhaul. At one extreme of this range of possibilities,
the ATs perform just analog filtering and (possibly) frequency conversion, the wired link are analog (e.g., radio over
fiber \cite{Niiho}), and the CP performs demodulation to baseband, A/D and D/A conversion, joint decoding (uplink)
and joint pre-coding (downlink). At the other extreme we have ``small cell'' architectures where the ATs perform encoding/decoding,
the wired links send data packets, and the CP performs high-level functions, such as scheduling, link-layer error control,
and macro-diversity packet selection.

\begin{figure}
\centerline{\includegraphics[width=11cm]{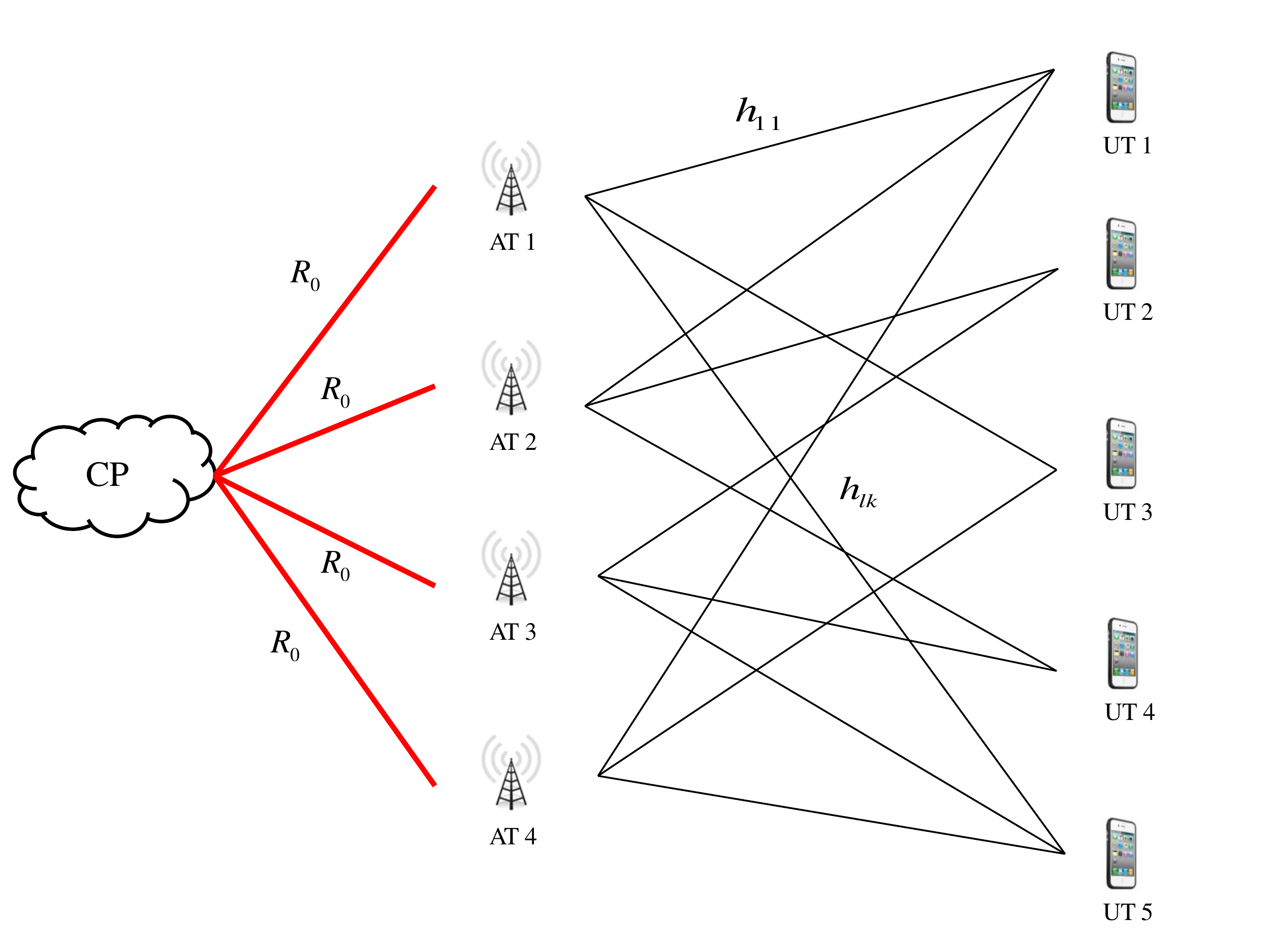}}
\caption{Distributed Antenna System with $5$ UTs and $4$ ATs (e.g., $K=5$ and $L=4$), and digital backhaul links of rate $R_{0}$.}
\label{DASModel}
\end{figure}

In this paper we focus on an intermediate  DAS architecture where the ATs perform partial decoding
(uplink) or precoding (downlink) and the backhaul is formed by digital links of fixed rate $R_0$.
In this case, the DAS uplink is an instance of a multi-source single destination layered relay network where the first layer is formed by the user terminals (UTs),
the second layer is formed by the ATs and the third layer contains just the CP (see Fig.~\ref{DASModel}). The corresponding DAS downlink  is an instance
of a broadcast layered relay network with independent messages.

In our model, analog forwarding from ATs to CP (uplink) or from CP to ATs (downlink) is not possible.
Hence, some form of quantization and forwarding is needed. A general approach to the uplink is based on the
{\em Quantize reMap and Forward} (QMF) paradigm of \cite{Avestimehr} (extended in \cite{Lim} where it is referred to as {\em Noisy Network Coding}).
In this case, the ATs perform vector quantization of their received signal at some rate
$R' \geq R_0$. They map the blocks of $nR'$ quantization bits into binary  words of length $nR_0$ by using
some randomized hashing function (notice that this corresponds to binning if $R' > R_0$), and let the CP perform joint decoding
of all UTs' messages based on the observation of all the (hashed) quantization bits.\footnote{The information-theoretic vector quantization
of \cite{Avestimehr}, \cite{Lim} can be replaced by scalar quantization with a fixed-gap performance degradation \cite{Ozgur}.}
It is known \cite{Avestimehr} that QMF achieves a rate region within a bounded gap from the cut-set outer bound \cite{Cover}, where the bound
depends only on the network size and on $R_0$, but it is independent of the channel coefficients and of the operating SNR.
For the broadcast-relay downlink, a general coding strategy has been proposed in \cite{Kannan} based on a combination of Marton coding
for the general broadcast channel \cite{Marton} and a coding scheme for deterministic linear relay networks, ``lifted'' to the Gaussian case.
Specializing the above general coding schemes to the the DAS considered here, for the uplink
we obtain the scheme based on quantization, binning and joint decoding of \cite{Sanderovich}, and for the downlink
we obtain the {\em Compressed Dirty-Paper Coding} (CDPC) scheme of \cite{Simeone}. From an implementation viewpoint, both QMF and CDPC are
not practical, the former requiring vector quantization at the ATs and {\em joint decoding} of all UT messages  based on the hashed quantization bits at the CP,
and the latter requiring {\em Dirty-Paper Coding} (notoriously difficult to implement in practice) and vector quantization at the CP.

A lower complexity alternative strategy for general relay networks was proposed in \cite{Nazer2011} and goes under the name of
{\em Compute and Forward} (CoF). CoF makes use of lattice codes, such that each relay can reliably decode a linear combination with
integer coefficient of the interfering codewords. Thank to the fact that lattices are modules over the ring of integers,
this linear combination translates directly into a linear combination of the information messages defined over a suitable finite field.
CoF can be immediately used for the DAS uplink. The performance of CoF was examined in \cite{Nazer2009} for the DAS uplink in the case of the
overly simplistic {\em Wyner model} \cite{Wyner}. It was shown that CoF yields competitive performance with respect to QMF for
practically realistic values of SNR.

This paper contributes to the subject in the following ways: 1) for the DAS uplink , we consider the CoF approach and
examine the corresponding system optimization at finite SNR for a general channel model including fading and shadowing (i.e., beyond the
nice and regular structure of the Wyner model); 2) For the downlink, we propose a novel precoding scheme nicknamed
{\em Reverse Compute and Forward} (RCoF); 3) For both uplink and downlink, we present low-complexity versions of CoF and RCoF
based on standard  scalar quantization at the receivers. These schemes are motivated by the observation that the main bottleneck of a digital receiver is the {\em Analog to Digital Conversion} (ADC), which is costly, power-hungry, and does not scale with Moore's law.
Rather the number of bit per second produced by an ADC is roughly a constant that depends on the power consumption \cite{Walden,Singh}.
Therefore, it makes sense to consider the ADC as part of the channel.  The proposed schemes, nicknamed {\em Quantized} CoF (QCoF) and {\em Quantized} RCoF (RQCoF),  lead to discrete-input discrete-output symmetric memoryless channel models naturally matched to standard single-user linear coding.
In fact, QCoF and RQCoF can be easily implemented using $q$-ary Low-Density Parity-Check (LDPC) codes \cite{songnamITW,Tunali,Feng}
with $q = p^2$ and $p$ prime, yielding essentially linear complexity in the code block
length and polynomial complexity in the system size (minimum between number of ATs and UTs).

The two major impairments that deteriorate the performance of DAS with CoF/RCoF are the non-integer penalty (i.e., the residual self-interference due to the fact that
the channel coefficients take on non-integer values in practice) and the rank-deficiency of the resulting system matrix over the
$q$-ary finite field. In fact, the wireless channel is characterized by fading and shadowing.
Hence, the channel matrix from ATs to UTs does not have any particularly nice structure, in contrast to the Wyner model case, where the channel matrix is tri-diagonal \cite{Nazer2009}. Thus, in a realistic setting, the system matrix resulting from CoF/RCoF may be rank deficient.  This is especially relevant when the size $q$
of the finite field is small (e.g., it is constrained by the resolution of the A/D and D/A conversion).
The proposed system optimization counters the above two problems by considering {\em power allocation},
{\em network decomposition} and {\em antenna selection} at the receivers (ATs selection in the uplink and
UTs selection in the downlink).  We show that in most practical cases the AT and UT selection problems can be optimally solved
by a simple greedy algorithm. Numerical results show that, in realistic networks with fading and
shadowing, the proposed optimization algorithms are very effective and essentially eliminate the problem of system matrix rank deficiency, even for small field size $q$.

A final novel contribution of this paper consists of the {\em Integer-Forcing Beamforming} (IFB) downlink scheme, targeted to the case where
$R_{0}$ is large, and therefore the DAS downlink  reduces to the well-known vector
Gaussian broadcast channel. In this case, a common and well-known low-complexity alternative to the capacity-achieving Gaussian DPC scheme
consists of Zero-Forcing Beamforming (ZFB), which achieves the same optimal multiplexing gain, at the cost of some
performance loss at finite SNR. IFB can be regarded both as a generalization of ZFB and
as the  dual of {\em Integer-Forcing Receiver} (IFR), proposed in \cite{Jiening} for the uplink multiuser MIMO case.
We demonstrate that IFB can achieve rates close to the information-theoretic optimal Gaussian  DPC, and can significantly outperform
conventional ZFB. This gain can be explained by the fact that IFB is able to reduce the power penalty of ZFB, due to
non-unitary beamforming.

The paper is organized as follows. In Section \ref{sec:Pre} we define the uplink and downlink DAS system model,
summarize some definitions on lattices and lattice coding, and review CoF.
In Section \ref{sec:DASCoF} we consider the application of CoF to the DAS uplink
and introduce the (novel) concept of {\em network decomposition} to improve the CoF sum rate.
Section \ref{sec:DASRCoF} considers the DAS downlink  and presents the RCoF scheme.
In Section \ref{sec:LowComp} we introduce the low-complexity ``quantized'' versions of CoF and RCoF.
Section \ref{sec:Wyner} focuses on the symmetric Wyner model and presents a simple {\em power allocation} strategy
to alleviate the impact of non-integer penalty.  In the case of a realistic DAS channel model
including fading, shadowing and pathloss,  a low-complexity greedy algorithm for ATs selection (uplink) and UTs
selection (downlink) is presented in Section \ref{sec:sch}. Finally, Section \ref{sec:IFB} considers the case of large backhaul rate
and presents the IFB scheme.  Some concluding remarks are provided in Section \ref{sec:conclusion}.

\section{Preliminaries}\label{sec:Pre}

In this section we provide some basic definitions and results that will be extensively used in the sequel.

\subsection{Distributed Antenna Systems: Channel Model}\label{subsec:Model}

We consider a DAS with $L$ ATs and $K$ UTs, each of which is equipped with a single antenna.
The ATs are connected to the CP  via digital backhaul links of rate $R_{0}$ (see Fig.~\ref{DASModel}).
A block of $n$ channel uses of the discrete-time complex baseband uplink channel is described by
\begin{equation} \label{channel-up}
\underline{\Ym} = \Hm \underline{\Xm} + \underline{\Zm},
\end{equation}
where we use ``underline'' to denote matrices whose horizontal dimension (column index) denotes  ``time'' and vertical dimension (row index) runs
across the antennas (UTs or ATs), the matrices
\[ \underline{\Xm} = \left [ \begin{array}{c} \underline{\xv}_1 \\ \vdots \\ \underline{\xv}_K \end{array} \right ] \;\;\; \mbox{and} \;\;\;
\underline{\Ym} = \left [ \begin{array}{c} \underline{\yv}_1 \\ \vdots \\ \underline{\yv}_L \end{array} \right ] \]
contain, arranged by rows, the UT codewords $\underline{\xv}_k \in \CC^{1 \times n}$ and
the AT channel output vectors $\underline{\yv}_\ell \in \CC^{1 \times n}$, for
$k = 1,\ldots, K$, and $\ell = 1,\ldots, L$, respectively.
The matrix $\underline{\Zm}$ contains i.i.d. Gaussian noise samples $\sim \Cc\Nc(0,1)$, and the matrix
$\Hm = [ \hv_1, \ldots, \hv_L]^\transp \in \CC^{L \times K}$ contains the channel coefficients,
assumed to be constant over the whole block of length $n$ and known to all nodes.

Similarly, a block of $n$ channel uses of the discrete-time complex baseband downlink channel is described by
$\underline{\tilde{\Ym}} = \tilde{\Hm} \underline{\tilde{\Xm}} + \underline{\tilde{\Zm}}$,
where we use ``tilde'' to denote downlink variables,  $\underline{\tilde{\Xm}} \in \CC^{L \times n}$
contains the AT codewords, $\underline{\tilde{\Ym}}, \underline{\tilde{\Zm}} \in \CC^{K \times n}$ contain the channel output and Gaussian noise
at the UT receivers, and $\tilde{\Hm} = [ \tilde{\hv}_1, \ldots, \tilde{\hv}_K]^\transp \in \CC^{K \times L}$  is the downlink channel matrix.

Since ATs and UTs are separated in space and powered independently, we assume a symmetric per-antenna power
constraint for both the uplink and the downlink, given by
$\frac{1}{n} \EE[\|\underline{\xv}_{k}\|^2]  \leq \SNR$ for all $k$ and by
$\frac{1}{n} \EE[\|\underline{\tilde{\xv}}_{\ell}\|^2]  \leq \SNR$ for all $\ell$, respectively.

\subsection{Nested Lattice Codes}\label{subsec:NLC}

Let $\ZZ[j]$ be the ring of Gaussian integers and $p$ be a Gaussian prime.~\footnote{A Gaussian integer is called a {\em Gaussian prime} if it is a prime in $\ZZ[j]$. A Gaussian prime $a+jb$ satisfies exactly one of the following conditions \cite{Stillwell}: 1) $|a|=|b|=1$; 2) one of $a,b$ is zero and the other is a prime number in $\ZZ$ of the form $4n+3$ (with $n$ a nonnegative integer); 3) both of $a,b$ are nonzero and $a^2 + b^2$ is a prime number in $\ZZ$ of the form $4n+1$. In this paper, $p$ is assumed to be a prime number congruent to $3$ modulo $4$, which is an {\em integer} Gaussian prime according to condition 2)..}
Let $\oplus$ denote the addition over $\FF_{p^2}$, and let $g: \FF_{p^2} \rightarrow \CC$ be the natural mapping of
$\FF_{p^2}$ onto $\{a+jb: a,b \in \ZZ_{p}\} \subset \CC$.
We recall the nested lattice code construction given in \cite{Nazer2011}.
Let $\Lambda = \{ \underline{\lambdav} = \underline{\zv} \Tm : \underline{\zv} \in \ZZ^n[j]\}$ be a lattice in $\CC^n$,
with full-rank generator matrix $\Tm \in \CC^{n \times n}$.
Let $\Cc  = \{ \underline{\cv} = \underline{\wv} \Gm : \wv \in \FF_{p^2}^r \}$ denote a linear code over $\FF_{p^2}$ with block length $n$ and dimension $r$, with generator matrix $\Gm$. The lattice $\Lambda_1$ is defined through ``construction A'' (see \cite{Erez2004} and references therein) as
\begin{equation} \label{construction-A}
\Lambda_1 = p^{-1} g(\Cc) \Tm + \Lambda,
\end{equation}
where $g(\Cc)$ is the image of $\Cc$ under the mapping $g$ (applied component-wise).
It follows that $\Lambda \subseteq \Lambda_1 \subseteq p^{-1} \Lambda$ is a chain of nested lattices, such that
$|\Lambda_1/\Lambda| = p^{2r}$ and $|p^{-1} \Lambda/\Lambda_1| = p^{2(n - r)}$.

For a lattice $\Lambda$ and $\underline{\rv} \in \CC^n$,
we define the lattice quantizer $Q_{\Lambda}(\underline{\rv}) = \argmin_{\underline{\lambdav} \in \Lambda}\|\underline{\rv} - \underline{\lambdav} \|^2$,
the Voronoi region $\Vc_\Lambda = \{\underline{\rv} \in \CC^{n}: Q_{\Lambda}(\underline{\rv}) = \underline{\zerov}\}$
and $[\underline{\rv}] \mod \Lambda = \underline{\rv} - Q_{\Lambda}(\underline{\rv})$.
For $\Lambda$ and $\Lambda_1$ given above, we define the lattice code
$\Lc = \Lambda_{1} \cap \Vc_\Lambda$ with rate $R = \frac{1}{n} \log |\Lc| = \frac{2r}{n}\log{p}$.
Construction A provides a {\em natural labeling}
of the codewords of $\Lc$ by the information messages $\underline{\wv} \in \FF^r_{p^2}$.  Notice that the set $p^{-1} g(\Cc)\Tm$ is a {\em system of coset representatives}
of the cosets  of $\Lambda$ in $\Lambda_1$. Hence, the natural labeling function  $f : \FF_2^r \rightarrow \Lc$
is defined by $f(\underline{\wv}) = p^{-1} g(\underline{\wv} \Gm)\Tm \mod \Lambda$.

\subsection{Compute and Forward} \label{subsec:CoF}

We recall here the CoF scheme of \cite{Nazer2011}.
Consider the $K$-user Gaussian multiple access channel (G-MAC) defined by
\begin{equation} \label{GMAC}
\underline{\yv} = \sum_{k = 1}^{K} h_{k} \underline{\xv}_{k} + \underline{\zv},
\end{equation}
where $\hv = [h_{1},\ldots,h_{K}]^\transp$, and the elements of $\underline{\zv}$ are i.i.d. $\sim \Cc\Nc(0,1)$.
All users make use of the same nested lattice codebook $\Lc = \Lambda_1 \cap \Vc_\Lambda$,
where $\Lambda$ has {\em second moment} $\sigma_\Lambda^2 \eqdef \frac{1}{n \mbox{Vol}(\Vc)} \int_{\Vc} \| \underline{\rv} \|^2 d\underline{\rv} = \SNR$.
Each user $k$ encodes its information message $\underline{\wv}_{k} \in \FF_{p^2}^{r}$ into the corresponding codeword $\underline{\tv}_{k} = f(\underline{\wv}_k)$
and produces its channel input according to
\begin{equation}\label{eq:channelinput}
\underline{\xv}_{k} = \left [\underline{\tv}_{k} + \underline{\dv}_{k} \right ] \mod \Lambda,
\end{equation}
where the {\em dithering sequences} $\underline{\dv}_{k}$'s are mutually independent across the users, uniformly distributed
over $\Vc_\Lambda$, and known to the receiver. Notice that, as in many other applications of nested lattice coding and lattice decoding (e.g.,\cite{Zamir,Erez2004,Damen}), random dithering is instrumental for the information theoretic proofs, but
a deterministic dithering sequence that scrambles the input and makes it zero-mean and uniform over the shaping region
can be effectively used in practice, without need of common randomness.  The decoder's goal is to recover a
linear combination $\underline{\vv} = [\sum_{k =1}^{K} a_{k} \underline{\tv}_{k}] \mod \Lambda$ with {\em integer coefficient vector}
$\av = [a_1,\ldots, a_K]^\transp \in \ZZ^K[j]$.
Since $\Lambda_1$ is a $\ZZ[j]$-module (closed under linear combinations with Gaussian integer coefficients),
then $\underline{\vv} \in \Lc$.  Letting $\hat{\underline{\vv}}$ be decoded codeword (for some decoding function which
in general depends on $\hv$ and $\av$), we say that a computation rate $R$ is achievable for this setting if there exists sequences
of lattice codes  $\Lc$ of rate $R$ and increasing block length $n$, such that the decoding error probability satisfies
$\lim_{n \rightarrow \infty} \PP(\hat{\underline{\vv}} \neq \underline{\vv} ) = 0$.

In the scheme of \cite{Nazer2011}, the receiver computes
\begin{eqnarray}
\hat{\underline{\yv}} & = & \left [\alpha \underline{\yv} - \sum_{k=1}^{K}a_{k} \underline{\dv}_{k} \right ] \mod \Lambda \nonumber \\
&=& \left [\underline{\vv} + \underline{\zv}_{\mbox{\tiny{eff}}} (\hv, \av,\alpha) \right ] \mod \Lambda, \label{eq:received}
\end{eqnarray}
where
\begin{equation}
\underline{\zv}_{\mbox{\tiny{eff}}} (\hv, \av, \alpha) = \sum_{k=1}^{K} (\alpha h_{k} - a_{k}) \underline{\xv}_{k} + \alpha \underline{\zv} \label{eq:effnoise}
\end{equation}
denotes the {\em effective noise}, including the non-integer self-interference (due to the fact that $\alpha h_k \notin \ZZ[j]$ in general)
and the additive Gaussian noise term.  The scaling, dither removal and modulo-$\Lambda$ operation in (\ref{eq:received})
is referred to as the {\em CoF receiver mapping} in the following.  By minimizing the variance of $\underline{\zv}_{\mbox{\tiny{eff}}}(\hv,\av,\alpha)$
with respect to $\alpha$, we obtain
\begin{eqnarray}
\sigma^{2} (\hv,\av) &=& \min_{\alpha} \sigma_{z_{\mbox{\tiny{eff}}}}^{2}(\hv,\av,\alpha) \nonumber \\
&=& \SNR\Big(\|\av\|^2 - \frac{\SNR|\hv^{\herm}\av|^2}{1+\SNR\|\hv\|^2}\Big) \nonumber \\
&\stackrel{(a)}{=}& \av^{\herm}(\SNR^{-1}\Id+\hv\hv^{\herm})^{-1}\av\label{eq:effnoise1}
\end{eqnarray}
where $(a)$ follows from the matrix inversion lemma \cite{Boyd}.
Since $\alpha$ is uniquely determined by $\hv$ and $\av$, it will be omitted in the following,
for the sake of notation simplicity.
 From \cite{Nazer2011}, we know that by applying lattice decoding to $\hat{\underline{\yv}}$ given in (\ref{eq:received}) the following computation rate is achievable:
\begin{equation}\label{eq:CoFrate}
R(\hv,\av,\SNR) = \log^{+}\Big(\frac{\SNR}{\av^{\herm}(\SNR^{-1}\Id+\hv\hv^{\herm})^{-1}\av}\Big),
\end{equation}
where $\log^{+}(x) \triangleq \max\{\log(x),0\}$.

The computation rate $R(\hv,\av,\SNR)$ can be maximized by minimizing $\sigma^{2} (\hv,\av)$ with respect to  $\av$.
The quadratic form (\ref{eq:effnoise1}) is positive definite for any $\SNR < \infty$, since the matrix
$(\SNR^{-1}\Id + \hv\hv^{\herm})^{-1}$ has eigenvalues
\begin{equation}\label{eq:eigen}
\lambda_{i} = \left \{
	\begin{array}{ll}
	\SNR/(1+\|\hv\|^2 \SNR) & \hbox{$i = 1$}\\
	\SNR & \hbox{$i> 1$}
	\end{array}
\right.
\end{equation}
By Cholesky decomposition, there exists a lower triangular matrix $\Lm$ such that $\sigma^{2} (\hv,\av)=\|\Lm^{\herm}\av\|^2$.
It follows that the problem of minimizing $\sigma^{2}(\hv,\av)$ over $\av \in \ZZ^{K}[j]$
is equivalent to finding the ''shortest lattice point" of the $L$-dimensional lattice generated by $\Lm^{\herm}$.
This can be efficiently obtained using the complex LLL algorithm \cite{LLL,Napias}\, possibly
followed by Phost or Schnorr-Euchner enumeration (see \cite{Damen}) of the non-zero lattice points in a sphere centered at the origin,
with radius equal to the shortest vector found by complex LLL.
Algorithm $1$ summarizes the procedures used in this paper to find the optimal
integer vector $\av \in \ZZ^{K}[j]$.

\begin{algorithm}\label{alg:lattice}
\caption{ Find the optimal integer coefficients}
\begin{description}
\item[1.] Take $\Fm = \Lm^{\herm}$
\item[2.] Find the reduced basis matrix $\Fm_{\texttt{\textrm{red}}}$, using the (complex) LLL algorithm
\item[3.] Take the column of $\Fm_{\texttt{\textrm{red}}}$ with minimum Euclidean norm, call it $\bv^{\star}$
\item[4.] Let $\rho=\|\bv^{\star}\| + \epsilon$ for some very small $\epsilon > 0$
\item[5.] Use Phost or Schnorr-Euchner enumeration  with $\Fm_{\texttt{\textrm{red}}}$ to find all lattice points in the sphere centered at $0$, with
    radius $\rho$.
\end{description}
Notice that this algorithm will find for sure the point $0$ (discarded), the point $\bv^{\star}$, and possibly some shorter non-zero points.
\end{algorithm}

\begin{figure}
\centerline{\includegraphics[width=11cm]{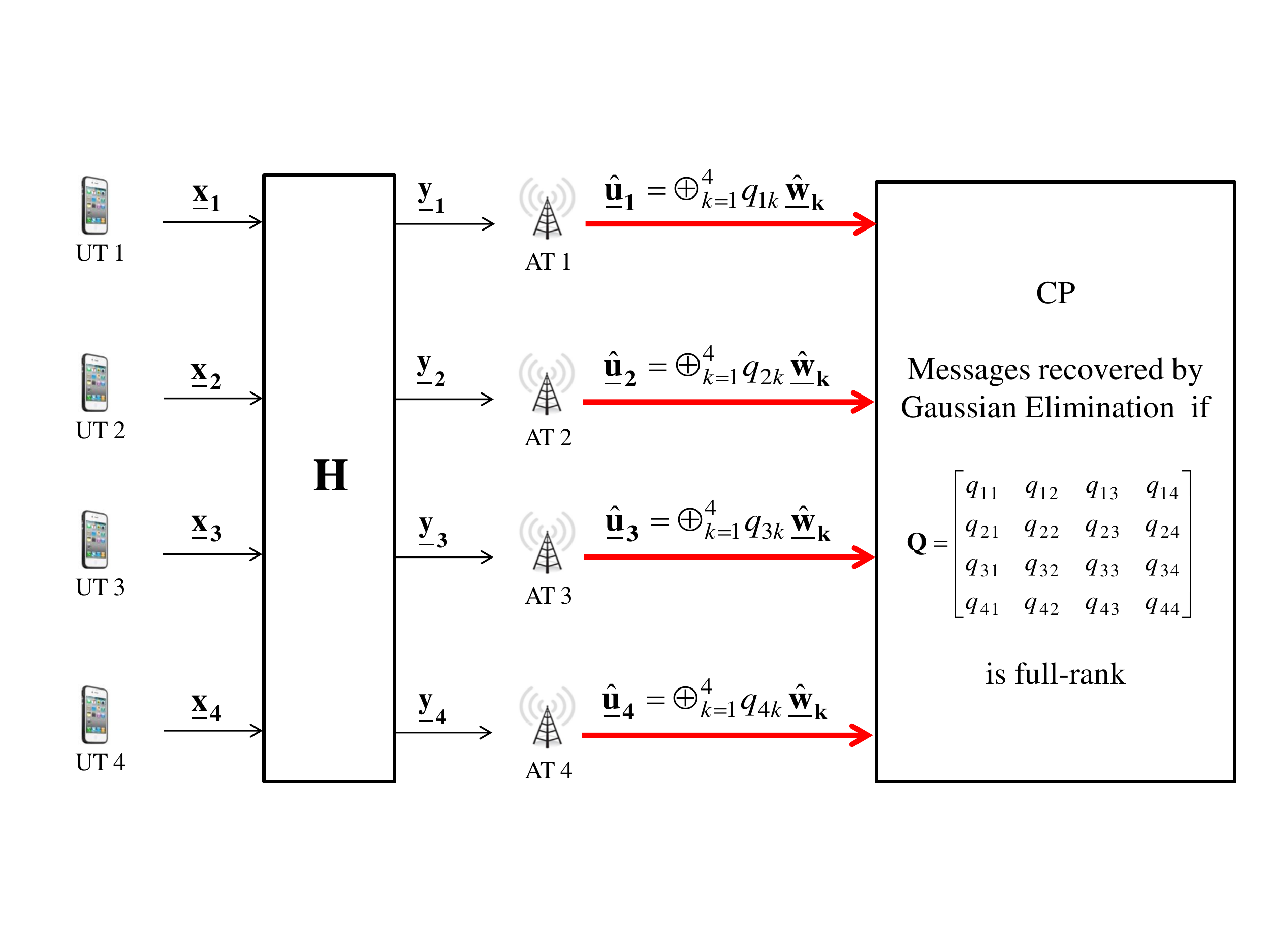}}
\caption{DAS Uplink Architecture using Compute and Forward: $L=4$ and $K=4$.}
\label{CoFDAS}
\end{figure}

\section{Compute and Forward for the DAS Uplink}\label{sec:DASCoF}

In this section we apply CoF to the DAS uplink and further improve its sum rate by introducing the idea of {\em network decomposition}.
The scheme is illustrated in Fig.~\ref{CoFDAS}, where CoF is used at each AT receiver.
For simplicity of exposition, we restrict to consider the same number $K = L$ of UTs and ATs. The notation, however,
applies also to the case of $K < L$ addressed in Section \ref{sec:sch}, when considering AT selection.
The UTs make use of the same lattice code $\Lc$ of rate $R$, and produce their channel input $\underline{\xv}_k$, $k = 1, \ldots, K$,
according to (\ref{eq:channelinput}).  Each AT $\ell$ decodes the codeword linear combination
$\underline{\vv}_{\ell} = \left [\sum_{k=1}^{K} a_{\ell,k}\underline{\tv}_{k} \right ] \mod \Lambda$,
for a target integer vector $\av_{\ell} = (a_{\ell,1},\ldots,a_{\ell,K})^\transp \in \ZZ^{K}[j]$ determined according to
Algorithm 1, independently of the other ATs.   If $R \leq R(\hv_{\ell},\av_{\ell},\SNR)$,  where the latter denotes
the computation rate of the G-MAC formed by the UTs and the $\ell$-th AT,
taking on the form given in (\ref{eq:CoFrate}),  the decoding error probability at AT $\ell$ can be made as small as desired.
Letting $\underline{\uv}_\ell = f^{-1}(\underline{\vv}_\ell)$ denote the information message corresponding to the target decoded codeword
$\underline{\vv}_\ell$, the code linearity over $\FF_{p^2}$ and the $\ZZ[j]$-module structure of $\Lambda_1$ yield
\begin{equation}\label{eq:isomorp}
\underline{\uv}_{\ell} = \bigoplus_{k = 1}^{K} q_{\ell,k} \underline{\wv}_{k},
\end{equation}
where $q_{\ell,k} = g^{-1}([a_{\ell,k}] \mod p \ZZ[j])$. After decoding, each AT $\ell$ forwards the corresponding information
message $\hat{\underline{\uv}}_{\ell}$ to the CP  via wired links of fixed $R_{0}$.
This can be done if $R \leq R_{0}$.  The CP collects all the messages $\hat{\underline{\uv}}_\ell$ for $\ell = 1,\ldots, L$ and forms the system of linear
equations over $\FF_{p^2}$
\begin{equation}
\left [ \begin{array}{c}
\hat{\underline{\uv}}_1 \\ \vdots \\ \hat{\underline{\uv}}_L \end{array} \right ] =
\Qm  \left [ \begin{array}{c}  \hat{\underline{\wv}}_1 \\ \vdots \\ \hat{\underline{\wv}}_K \end{array} \right ],
\end{equation}
where we define $\Am = [\av_1, \ldots, \av_L]^\transp$ and the system matrix
$\Qm = [\qv_{1},\ldots,\qv_{L}]^{\transp} = g^{-1}\left ( [\Am] \mod p\ZZ[j]\right )$.
Provided that $\Qm$ has rank $K$ over $\FF_{p^2}$, the CP obtains the decoded messages $\{\hat{\underline{\wv}}_k\}$
by Gaussian elimination. Assuming this full-rank condition and $R < R(\hv_{\ell},\av_{\ell},\SNR)$ for all $\ell = 1,\ldots, L$,
the error probability $\PP(\hat{\underline{\wv}}_k \neq \underline{\wv}_k \;\; \mbox{for some} \; k)$ can be made arbitrarily small for sufficiently
large $n$.  The resulting achievable rate per user  is given by \cite{Nazer2009}:
\begin{eqnarray}\label{eq:sumrateCoF}
R = \min\{R_{0}, \min_{\ell} \{R(\hv_{\ell},\av_{\ell},\SNR)\} \}.
\end{eqnarray}

\begin{remark}
Since each AT $\ell$ determines  its coefficients vector $\av_{\ell}$ in a decentralized way, by applying
Algorithm 1 independently of the other ATs' channel coefficients, the resulting system matrix $\Qm$
may be rank-deficient. If $K < L$, requiring that all ATs can decode reliably is unnecessarily restrictive: it is sufficient to
select a subset of $K$ ATs which can decode reliably and whose coefficients form a full-rank system matrix.
This selection problem will be addressed in Section \ref{sec:sch}. \hfill $\lozenge$
\end{remark}

The sum rate of CoF-based DAS can be improved by {\em network decomposition} with respect to the system matrix $\Qm$.
Although the elements of $\Hm$ are non-zero,  the corresponding $\Qm$ may include zeros, since some elements of the vectors $\av_{\ell}$ may be
zero modulo $p\ZZ[j]$.  Because of the presence of zero elements, the system matrix $\Qm$ may be put in block diagonal form by column
and row permutations.
If the permuted system matrix has $S$ diagonal blocks,
the corresponding network graph decomposes into $S$ independent subnetworks
and CoF can be applied separately to each subnetwork such that taking the minimum of the computation rates over the
subnetworks is not needed. Hence, the sum rate is given by the sum (over the subnetworks) of the sum rates of each network component.
In turns, the common UT rate of each indecomposable subnetwork takes on the form (\ref{eq:sumrateCoF}).
For given $\Qm$, the disjoint subnetwork components can be found efficiently using depth-first or breadth-first search \cite{Ahuja}.
This also essentially reduces the computation complexity of Gaussian elimination, which is performed independently
for each subnetwork. We assume that, up to a suitable permutation of rows and columns,
$\Qm$ can be put in block diagonal form with diagonal blocks $\Qm(\Ac_{s},\Uc_{s})$ for $s  = 1, \ldots, S$, where we
use the following notation:  for a matrix $\Qm$ with rows index set $[1:L]$ and column index set $[1:K]$,
$\Qm(\Ac, \Uc)$ denotes the submatrix  obtained by selecting the rows in $\Ac \subseteq [1:L]$ and the columns in
$\Uc \subseteq [1:K]$. The following results are immediate:

\begin{lemma}\label{lem:decomp} If $\Qm$ is a full-rank $K \times K$ matrix, the diagonal blocks
$\Qm(\Ac_{s},\Uc_{s})$ are full-rank square matrices for every $s$. \hfill \IEEEQED
\end{lemma}

\begin{theorem}\label{thm:CoF}
CoF with network decomposition, applied to a DAS uplink with channel matrix $\Hm=[\hv_{1},\ldots,\hv_{K}]^{\transp} \in \CC^{K \times K}$, achieves the sum rate
\begin{eqnarray}
R_{\mbox{\tiny{CoF}}}(\Hm,\Am) &=& \sum_{s=1}^{S} |\Ac_{s}| \min\left \{R_{0},\min \{R(\hv_{k},\av_{k},\SNR): k \in \Ac_{s} \} \right \},
\end{eqnarray}
where $\Am=[\av_{1},\ldots,\av_{K}]^{\transp}$ is the matrix of CoF integer coefficients, and where the system matrix
$\Qm=g^{-1}([\Am] \mod p\ZZ[j])$ has full rank $K$ over $\FF_{p^2}$ and can be put in block diagonal form by rows
and columns permutations, with diagonal blocks  $\Qm(\Ac_{s},\Uc_{s})$ for $s=1,\ldots,S$. \hfill \IEEEQED
\end{theorem}

\section{Reverse Compute and Forward for the DAS Downlink}\label{sec:DASRCoF}

\begin{figure}
\centerline{\includegraphics[width=11cm]{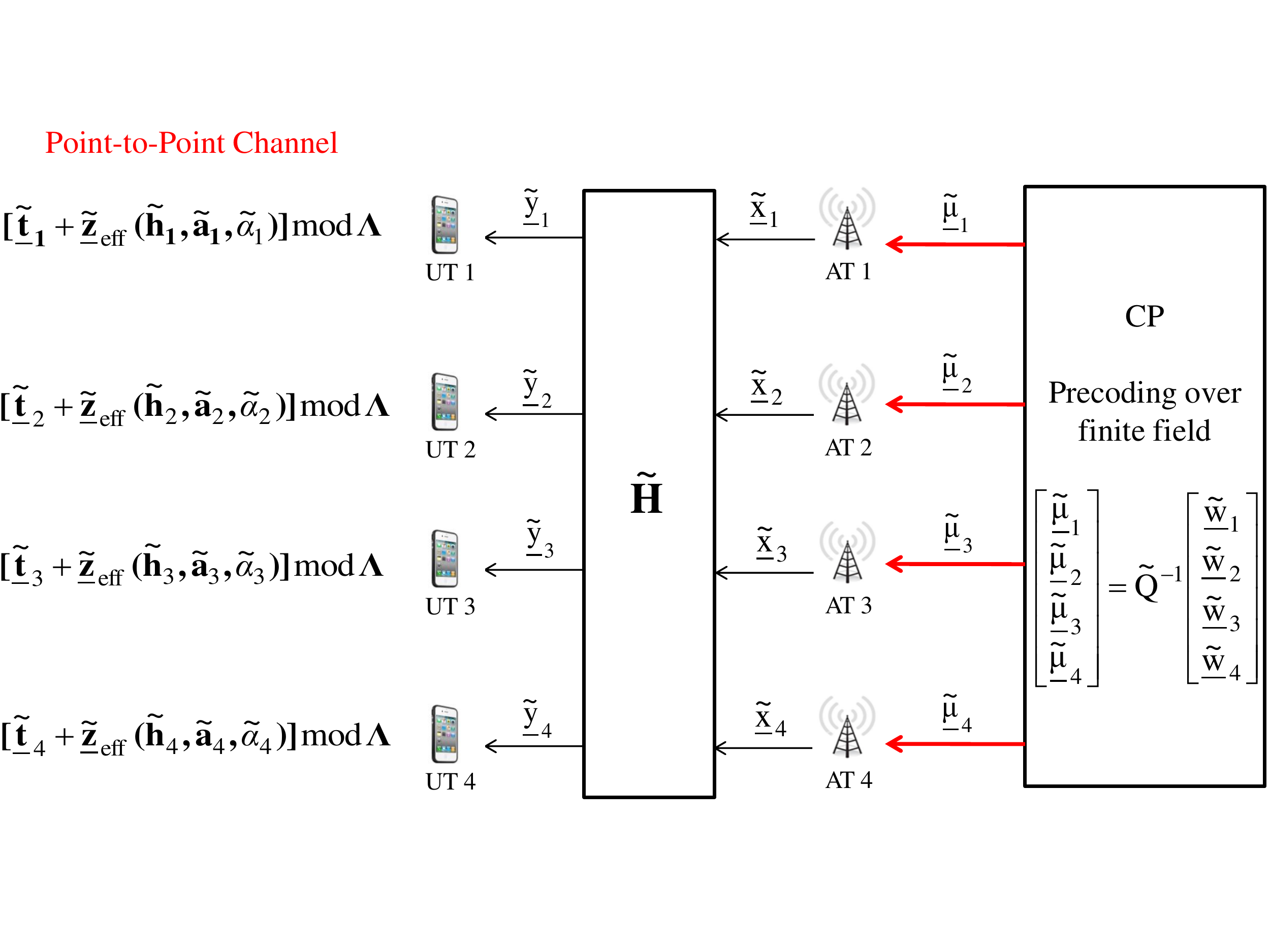}}
\caption{DAS Downlink  Architecture Using Reverse Compute and Forward: $L=4$ and $K=4$.}
\label{RCoFDAS}
\end{figure}

In this section we propose a novel downlink precoding scheme nicknamed ``Reverse" CoF (RCoF).
Again, we restrict to the case $K = L$ although the notation applies to the
case of $K > L$, treated in Section \ref{sec:sch}.
In a DAS downlink, the role of the ATs and UTs can be reversed with respect to the uplink.  Each UT can reliably decode an integer linear
combination of the lattice codewords sent by the ATs.  However, the UTs cannot share the decoded codewords as in the uplink,
since they have no backhaul links. Instead, the ``interference'' in the finite-field domain can be totally eliminated by zero-forcing
precoding (over the finite field) at the CP.  RCoF has a distinctive advantage with respect to its CoF counterpart viewed before:
since each UT sees only its own lattice codeword plus the effective noise, each message is rate-constrained
by the computation rate of its own intended receiver, and not by the minimum of all computation rates across all receivers,
as in the uplink case.  In order to achieve different coding rates while preserving the lattice $\ZZ[j]$-module structure,
we use a {\em family} of nested lattices $\Lambda \subseteq \Lambda_{L} \subseteq \cdots \subseteq \Lambda_{1}$,
obtained by a nested construction  A as described in \cite[Sect. IV.A]{Nazer2011}.
In particular, we let $\Lambda_\ell = p^{-1} g(\Cc_\ell) \Tm + \Lambda$ with $\Lambda = \ZZ^n[j] \Tm$ and
with $\Cc_\ell$ denoting the linear code over $\FF_{p^2}$ generated by the first $r_\ell$ rows of a common generator matrix $\Gm$,
with $r_L \leq r_{L-1} \leq \cdots \leq r_1$.  The corresponding nested lattice codes are given by $\Lc_{\ell} = \Lambda_{\ell} \cap \Vc_\Lambda$,
and have rate $R_\ell = \frac{2r_\ell}{n} \log p$.  We let $\tilde{\Am} = [\tilde{\av}_{1},\ldots,\tilde{\av}_{K}]^{\transp}$,
where  $\tilde{\av}_{k} \in \ZZ^{L}[j]$ denotes the integer coefficients vector used at UT $k$ for the modulo-$\Lambda$ receiver mapping
(see (\ref{eq:received})), and we let $\tilde{\Qm} = g^{-1}( [ \tilde{\Am} ]  \mod p\ZZ[j])$ denote the downlink
system matrix, assumed to have rank $L$.  Then, RCoF scheme proceeds as follows (see Fig.~\ref{RCoFDAS}):
\begin{itemize}
\item The CP sends $L$ independent messages to $L$ UTs (if $K > L$, then a subset of $L$ UTs is selected, as explained in Section \ref{sec:sch}).
We let $k_\ell$ denote the UT destination of the $\ell$-th message, encoded by $\Lc_\ell$ at rate $R_\ell$.
\item The CP forms the messages $\tilde{\underline{\wv}}_\ell \in \FF_{p^2}^{r_1}$ by appending $r_1 - r_\ell$ zeros
to each $\ell$-th information message of $r_\ell$ symbols, so that all messages have the same length $r_1$.
\item The CP produces the precoded messages
\begin{equation}\label{eq:precod}
\left [ \begin{array}{c}
\underline{\tilde{\muv}}_{1} \\ \vdots \\ \underline{\tilde{\muv}}_{L} \end{array} \right ] = \tilde{\Qm}^{-1} \left [ \begin{array}{c}
\underline{\tilde{\wv}}_{1} \\ \vdots \\ \underline{\tilde{\wv}}_{L} \end{array} \right ].
\end{equation}
(notice: if $K > L$ then $\tilde{\Qm}$ is replaced by the $L \times L$ submatrix
$\tilde{\Qm}(\{k_1,\ldots, k_L\}, [1:L])$).
\item  The CP forwards the precoded message $\underline{\tilde{\muv}}_{\ell}$  to  AT $\ell$ for all $\ell = 1,\ldots, L$, via the digital
backhaul link.
\item AT $\ell$ locally produces the lattice codeword $\underline{\nuv}_\ell = f(\underline{\tilde{\muv}}_{\ell}) \in \Lc_1$
(the densest  lattice code)  and transmits the corresponding channel input $\underline{\tilde{\xv}}_{\ell}$ according to
(\ref{eq:channelinput}). Because of linearity, the precoding and the encoding over the finite field commute. Therefore,
we can write $[\underline{\tilde{\nuv}}^\transp_{1},\ldots, \underline{\tilde{\nuv}}^\transp_{L}]^{\transp} = \Bm [\underline{\tilde{\tv}}^\transp_{1},\ldots,
\underline{\tilde{\tv}}^\transp_{L}]^{\transp} \mod \Lambda$, where $\underline{\tilde{\tv}}_{\ell} = f(\underline{\tilde{\wv}}_\ell)$
and $\Bm = g( \tilde{\Qm}^{-1})$.
\item Each UT $k_\ell$ applied the CoF receiver mapping as in (\ref{eq:received}), with
integer coefficients vector $\tilde{\av}_{k_\ell}$ and scaling factor $\alpha_{k_\ell}$, yielding
\begin{eqnarray}
\hat{\underline{\tilde{\yv}}}_{k_\ell} & = & \left [\tilde{\av}_{k_\ell}^{\transp} \left [ \begin{array}{c}
\underline{\tilde{\nuv}}_{1} \\ \vdots \\ \underline{\tilde{\nuv}}_{L} \end{array} \right ] + \underline{\tilde{\zv}}_{\mbox{\tiny{eff}}} (\tilde{\hv}_{k_\ell}, \tilde{\av}_{k_\ell},\alpha_{k_\ell}) \right ] \mod \Lambda \nonumber \\
&=&  \left [ \tilde{\av}_{k_\ell}^{\transp} \Bm \left [\begin{array}{c}
\underline{\tilde{\tv}}_{1} \\ \vdots \\ \underline{\tilde{\tv}}_{L} \end{array} \right ]   +  \underline{\tilde{\zv}}_{\mbox{\tiny{eff}}}(\tilde{\hv}_{k_\ell}, \tilde{\av}_{k_\ell},\alpha_{k_\ell}) \right ] \mod \Lambda \nonumber \\
&\stackrel{(a)}{=}& \left [ \left ( \left [ \tilde{\av}_{k_\ell}^{\transp} \Bm \right ] \mod p\ZZ[j] \right ) \left [ \begin{array}{c}
\underline{\tilde{\tv}}_{1} \\ \vdots \\ \underline{\tilde{\tv}}_{L} \end{array} \right ] + \underline{\tilde{\zv}}_{\mbox{\tiny{eff}}}(\tilde{\hv}_{k_\ell}, \tilde{\av}_{k_\ell},\alpha_{k_\ell}) \right ] \mod \Lambda \nonumber \\
&\stackrel{(b)}{=}& \left [ \underline{\tilde{\tv}}_{\ell} + \underline{\tilde{\zv}}_{\mbox{\tiny{eff}}}(\tilde{\hv}_{k_\ell}, \tilde{\av}_{k_\ell},\alpha_{k_\ell}) \right ] \mod \Lambda,
\label{minchia}
\end{eqnarray}
where (a) is due to the fact that $[p \; \underline{\tv}] \mod \Lambda = \underline{\zerov}$ for any codeword
$\underline{\tv} \in \Lambda_\ell$,  and (b) follows from the following result:
\end{itemize}

\begin{lemma}\label{lem:inverse}
Let $\tilde{\Qm} = g^{-1}([\tilde{\Am}] \mod p\ZZ[j])$.
Assuming $\tilde{\Qm}$ invertible over $\FF_{p^2}$,
if $\Bm = g(\tilde{\Qm}^{-1})$, then:
\begin{equation}
[\tilde{\Am} \Bm] \mod p\ZZ[j] = \Id.
\end{equation}
\end{lemma}
\begin{proof} Using $[\tilde{\Am}] \mod p\ZZ[j] = g(\tilde{\Qm})$, we have:
\begin{eqnarray}
[\tilde{\Am}\Bm] \mod p\ZZ[j]&=&[ ( [\tilde{\Am}] \mod p\ZZ[j] ) \; \Bm] \mod p\ZZ[j]\\
&=& [g(\tilde{\Qm})g(\tilde{\Qm}^{-1})] \mod p\ZZ[j]\\
&=& [g(\tilde{\Qm}\tilde{\Qm}^{-1})] \mod p\ZZ[j] \\
& = & \Id.
\end{eqnarray}
\end{proof}

From (\ref{minchia}) we have that RCoF induces a point-to-point channel at each desired UT $k_\ell$,
where the the integer-valued interference is eliminated by precoding,  and the remaining effective noise is due to the non-integer residual interference
and to the channel Gaussian noise. The scaling coefficient $\alpha_{k_\ell}$  and the integer vector $\tilde{\av}_{k_\ell}$ are optimized independently by
each UT using (\ref{eq:effnoise1}) and Algorithm 1. It follows that  the desired message $\underline{\tilde{\wv}}_{\ell}$ can be recovered with arbitrarily small
probability  of error if $R_{\ell} \leq R(\tilde{\hv}_{k_\ell},\tilde{\av}_{k_\ell},\SNR)$, where the latter takes on the form given in (\ref{eq:CoFrate}).
Including the fact that the precoded messages can be sent from the CP to the ATs if  $R_1 \leq R_0$, we arrive at:

\begin{theorem}\label{thm:RCoF}
RCoF applied to a DAS downlink  with channel matrix $\tilde{\Hm}=[\tilde{\hv}_{1},\ldots,\tilde{\hv}_{L}]^{\transp} \in \CC^{L \times L}$ achieves the sum rate
\begin{eqnarray*}
R_{\mbox{\tiny{RCoF}}}(\tilde{\Hm},\tilde{\Am}) & = & \sum_{\ell = 1}^{L} \min\{R_{0}, R(\tilde{\hv}_{\ell},\tilde{\av}_{\ell},\SNR)\}.
\end{eqnarray*}
\begin{flushright}$\blacksquare$\end{flushright}
\end{theorem}

\begin{remark} When the channel matrix $\tilde{\Hm}$ has the property that each row $\ell$ is a permutation of the first row (e.g., in the case $\tilde{\Hm}$
is circulant, as in the Wyner model \cite{Wyner}), each UT has the same computation rate and hence a single lattice code
$\Lc = \Lc_1 = \cdots = \Lc_L$ is sufficient.
\hfill $\lozenge$
\end{remark}

\section{Low-Complexity Schemes}\label{sec:LowComp}

This section considers low-complexity versions of the schemes of
Sections \ref{sec:DASCoF} and \ref{sec:DASRCoF}, using one-dimensional lattices and scalar quantization.
Our approach is suited to the practically relevant case where the receivers are equipped with ADCs of fixed finite resolution,
such that scalar quantization is included as an unavoidable part of the channel model.
In this case, CoF and RCoF, as well as QMF and CDPC, are not possible since lattice quantization requires to have access to the
unquantized (soft) signal samples.

The quantized versions of CoF and RCoF follow as a special cases,
by choosing the generator matrix of the shaping lattice $\Lambda$ to be $\Tm = \tau \Id$, with
$\tau = \sqrt{6 \SNR}$  in order to satisfy the per-antenna power constraint with equality.
The resulting lattice code is $\Lc = \Lambda_1 \cap \Vc_\Lambda$ with
$\Lambda = \tau \ZZ^n[j]$ and $\Lambda_1 = (\tau/p) g(\Cc) + \Lambda$,
for a linear code $\Cc$ over $\FF_{p^2}$  of rate $R = \frac{2r}{n} \log p$.
Furthermore, we introduce a scalar quantization stage as part of each receiver.
This is defined by the function $Q_{(\tau/p) \ZZ[j]}(\cdot)$, applied component-wise.
Since $\Lambda$ is the $n$-dimensional complex cubic lattice, also the modulo-$\Lambda$
operations in CoF/RCoF are performed component-wise. Hence, we can restrict to a symbol-by-symbol channel model instead of considering
$n$-vectors as before.

Consider the same G-MAC setting of Section \ref{subsec:CoF}.  Given the information message
$\underline{\wv}_k \in \FF_{p^2}^r$, encoder $k$ produces the codeword $\underline{\cv}_k = \underline{\wv}_k \Gm$ and the corresponding lattice codeword
$\underline{\tv}_k = f(\underline{\wv}_k) =  (\tau/p) g(\underline{\cv}_k) \mod \Lambda$.
The $i$-th component of its channel input $\underline{\xv}_k$ is given by
\begin{equation}  \label{transmit-quant}
x_{k,i} = \left [ t_{k,i} + d_{k,i} \right ] \mod \tau \ZZ[j],
\end{equation}
where the dithering samples $d_{k,i}$ are i.i.d. across users and time dimensions, and uniformly distributed over the square region
$[0, \tau) + j[0,\tau)$.  The received signal is given by (\ref{GMAC}).
The receiver selects the {\em integer coefficients vector} $\av = (a_{1},\ldots,a_{K})^{\transp} \in \ZZ^{K}[j]$
and produces the sequence $\underline{\uv}  \in \FF_{p^2}^n$ with components
\begin{eqnarray}\label{eq:demo}
u_{i} & = &   g^{-1} \left( \frac{p}{\tau} \left ( \left[ Q_{(\tau/p) \ZZ[j]} \left ( \alpha y_{i}  -  \sum_{k = 1}^{K} a_k d_{k,i} \right) \right] \mod \tau \ZZ[j] \right ) \right)\\
&=&  g^{-1} \left( \left[ Q_{\ZZ[j]} \left ( \frac{p}{\tau} \left ( \sum_{\ell=1}^{K} a_k t_{k,i} +  \xi_{i}(\hv,\av,\alpha) \right ) \right) \right ] \mod p \ZZ[j] \right),
\end{eqnarray}
for $i=1,\ldots,n$, where
\begin{equation} \label{eff-noise-quant}
\xi_{i}(\hv,\av,\alpha) =  \sum_{k = 1}^{K}(\alpha h_{k} - a_{k}) x_{k,i} + \alpha z_{i}.
\end{equation}
Since $\frac{p}{\tau} t_{k,i} \in \ZZ[j]$ by construction,
and using the obvious identity $Q_{\ZZ[j]}(v + \xi) = v + Q_{\ZZ[j]}(\xi)$ with $v \in \ZZ[j]$ and $\xi \in \CC$,
we arrive at
\begin{equation}\label{eq:FFmodel}
\underline{\uv} = \Big( \bigoplus_{k = 1}^{K} q_{k} \underline{\cv}_{k} \Big) \oplus \underline{\zetav}(\hv,\av,\alpha),
\end{equation}
where $q_{k} = g^{-1}([a_{k}] \mod p\ZZ[j])$ and where the components
of the discrete additive noise  $\underline{\zetav}(\hv,\av,\alpha)$ are given by  $\zeta_i(\hv,\av,\alpha) = g^{-1}( [Q_{\ZZ[j]}((p/\tau) \xi_i(\hv,\av,\alpha))] \mod p\ZZ[j])$.
This shows that the concatenation of the lattice encoders, the G-MAC and the receiver mapping (\ref{eq:demo}) reduces to
an equivalent discrete linear additive-noise finite-field MAC (FF-MAC) given by (\ref{eq:FFmodel}).

\begin{figure}
\centerline{\includegraphics[width=11cm]{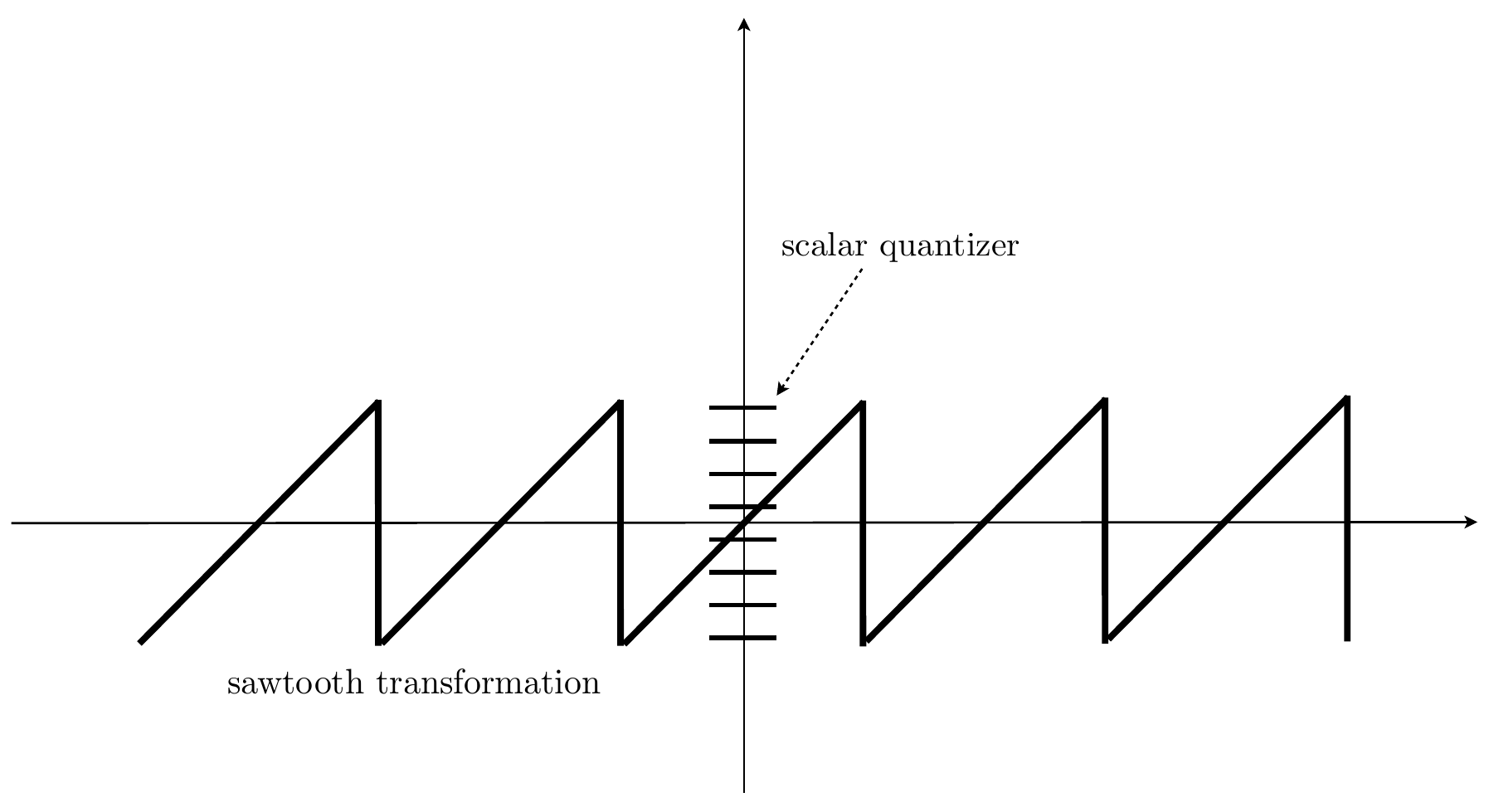}}
\caption{Implementation of the modulo $\Lambda$ operation (analog component-wise sawtooth transformation) followed by the scalar quantization
function $Q_{(\tau/p) \ZZ[j]}(\cdot)$ function.}
\label{receiver-structure}
\end{figure}


\begin{remark}
Notice that $\underline{\uv}$ is obtained from the channel output $\underline{\yv}$ by component-wise analog operations
(scaling by $\alpha$ and translation by $\sum_{k = 1}^{K} a_k \underline{\dv}_k$), scalar quantization and modulo $\Lambda$ reduction.
In fact, the scalar quantization and the modulo lattice operations commute, i.e., the modulo operation can be performed
directly on the analog signals by wrapping the complex plane into the Voronoi region of $\tau \ZZ[j]$, and then the scalar quantizer
$Q_{(\tau/p)\ZZ[j]}(\cdot)$ can be applied to the wrapped samples. This corresponds to the analog sawtooth transformation, followed by scalar quantization, applied
to the real and imaginary parts of the complex baseband signal,
as shown in Fig.~\ref{receiver-structure}.
\hfill $\lozenge$
\end{remark}

The marginal pmf of $\zeta_{i}(\hv,\av,\alpha)$ can be calculated numerically, and it is well approximated by assuming
$(p/\tau) \xi_{i}(\hv,\av,\alpha) \sim \Cc\Nc(0,\sigma_{\xi}^2)$. In Appendix~\ref{app:GA}, we obtain an accurate and easy way
to calculate the pmf of the effective noise component $\zeta_{i}(\hv,\av,\alpha)$ based on such Gaussian approximation.
The optimal choice of $\av$ and $\alpha$ for
the discrete channel (\ref{eq:FFmodel}) consists of minimizing the entropy of the discrete additive noise $H(\zeta_{i}(\hv,\av,\alpha))$.
However, this does not lead to a tractable numerical method. Instead, we resort to the minimization of the
unquantized effective noise variance $\sigma_{\xi}^2$, which leads to the same expression (\ref{eq:effnoise1})
and integer search of Algorithm 1. We assume that $\alpha$ and $\av$ are determined in this way, independently, by each receiver,
and  omit $\alpha$ from the notation.

In the following, we will present coding schemes for the induced FF-MAC in (\ref{eq:FFmodel}) and for the corresponding
Finite-Field Broadcast Channel (FF-BC) resulting from the downlink, by exchanging the roles of ATs and UTs.
We follow the notation used in Sections \ref{sec:DASCoF} and \ref{sec:DASRCoF} and let
$\Qm = g^{-1}([\Am] \mod p\ZZ[j])$ and  $\tilde{\Qm} = g^{-1}([\tilde{\Am}] \mod p\ZZ[j])$
denote the system matrix for the uplink and for the downlink, respectively.

\subsection{QCoF and LQF for the DAS Uplink}\label{subsec:QCoF}

In this section we present two schemes referred to as Quantized CoF (QCoF) and Lattice Quantize and Forward (LQF), which differ by the processing at the ATs.
QCoF is a low-complexity quantized version of CoF. The quantized channel output at AT $\ell$ is given by
\begin{equation}\label{eq:FFmodel1}
\underline{\uv}_{\ell} = \underline{\vv}_{\ell} \oplus \underline{\zetav}(\hv_{\ell},\av_{\ell}),
\end{equation}
where, by linearity,  $\underline{\vv}_{\ell} =\bigoplus_{k=1}^{K} q_{\ell,k} \underline{\cv}_{k}$ is a codeword of $\Cc$.
This is a point-to-point channel with discrete additive noise over $\FF_{p^2}$.
AT $\ell$ can successfully decode $\underline{\vv}_{\ell}$ if $R \leq 2\log{p} - H(\zeta(\hv_{\ell},\av_\ell))$. This is an immediate consequence of the well-known
fact that linear codes achieve the capacity of symmetric discrete memoryless channels\cite{Dobrushin}.
If $R \leq R_{0}$, each AT $\ell$ can forward the decoded message linear finite-field combination to the CP, so that the original UT messages
can be obtained by Gaussian elimination (see Section \ref{sec:DASCoF}). With the same notation of Theorem \ref{thm:CoF}, including network decomposition
which applies verbatim here, we have:

\begin{theorem}\label{thm:QCoF}
QCoF with {\em network decomposition}, applied to
a DAS uplink with channel matrix $\Hm=[\hv_{1},\ldots,\hv_{K}]^{\transp}\in \CC^{K \times K}$, achieves the sum rate
\begin{equation}
R_{\mbox{\tiny{QCoF}}}(\Hm,\Am) = \sum_{s=1}^{S} |\Ac_s| \min\left \{R_{0}, \min \{2\log{p} - H(\zeta(\hv_{k},\av_{k})): k \in \Ac_{s}\} \right \},
\end{equation}
\hfill \IEEEQED
\end{theorem}

Next, we consider the LQF scheme, which may provide an attractive alternative in the case $2\log{p} \leq R_{0}$, i.e., when $R_0$ is large and
a small value of $p$ is imposed by the ADC complexity and/or power consumption constraints.
In LQF, the UTs encode their information messages by using independently generated, not nested, random linear codes  $\{\Cc_{k}\}$ over $\FF_{q}$, in order to allow for different coding rates $\{R_{k}\}$. In this case, the fine lattice for UT $k$ is $\Lambda_{k} = (\tau/p)g(\Cc_{k}) + \tau\ZZ^{n}[j]$ and the symbol by symbol quantization maps the channel into an additive MAC channel over $\FF_q$, with discrete additive noise. Hence, independently generated random linear codes are optimal for this channel (this is easily seen form the fact that the channel is additive over the finite field).
In LQF, the ATs forwards its quantized channel observations directly to the CP without
local decoding. Hence, LQF can be seen as a special case of QMF without binning.
From (\ref{eq:FFmodel1}), the CP  sees a FF-MAC with $L$-dimensional output:
\begin{equation}\label{eq:FF-MAC}
\left[
  \begin{array}{c}
    \underline{\uv}_{1} \\
    \vdots \\
    \underline{\uv}_{L} \\
  \end{array}
\right] = \Qm \left[
                \begin{array}{c}
                  \underline{\cv}_{1} \\
                  \vdots \\
                  \underline{\cv}_{K} \\
                \end{array}
              \right] \oplus \left[
                \begin{array}{c}
                  \underline{\zetav}(\hv_{1},\av_{1}) \\
                  \vdots \\
                  \underline{\zetav}(\hv_{L},\av_{L}) \\
                \end{array}
              \right].
\end{equation}
The following result provides an achievable sum rate of LQF subject to the constraint $2 \log p \leq R_0$.
\begin{theorem}\label{lem:LQF}
Consider the FF-MAC, defined by $\Qm \in \FF_{p^2}^{K \times K}$ as in (\ref{eq:FF-MAC}).  If $\Qm$ has rank $K$, the following sum rate is achievable
by linear coding
\begin{equation} \label{sucacazzi}
R_{\mbox{\tiny{FF-MAC}}} = 2 K \log{p} - \sum_{k=1}^{K} H(\zeta(\hv_{k},\av_{k})).
\end{equation}
\end{theorem}
\begin{IEEEproof} See Appendix \ref{proof:FF-MAC}.
\end{IEEEproof}

The relative merit of QCoF and LQF depends on $R_{0}$, $p$, and on the actual realization of the channel matrix $\Hm$.
In symmetric channel cases (i.e., Wyner model \cite{Wyner}), where the AT have the same computation rate,
QCoF beats LQF by making $p$ sufficiently large.
On the other hand, if the modulation order $p$ is predetermined as in a conventional wireless communication
system, and this is relatively small with respect to $R_{0}$, LQF outperforms QCoF by breaking the limitation
of the minimum computation rate over the ATs.

\subsection{RQCoF for the DAS Downlink}\label{subsec:RQCoF}

Exchanging the roles of AT s and UTs and using (\ref{eq:FFmodel}), the DAS downlink with quantization at the receivers
is turned into the FF-BC
\begin{equation}\label{eq:FF-BC}
\left[
  \begin{array}{c}
    \underline{\tilde{\uv}}_{1} \\
    \vdots \\
    \underline{\tilde{\uv}}_{K} \\
  \end{array}
\right] = \tilde{\Qm}
\left[
                \begin{array}{c}
                  \underline{\tilde{\cv}}_{1} \\
                  \vdots \\
                  \underline{\tilde{\cv}}_{L} \\
                \end{array}
              \right] \oplus \left[
                \begin{array}{c}
                  \underline{\zetav}(\tilde{\hv}_{1},\tilde{\av}_{1}) \\
                  \vdots \\
                  \underline{\zetav}(\tilde{\hv}_{K},\tilde{\av}_{K}) \\
                \end{array}
              \right].
\end{equation}
The following result yields that simple matrix inversion over $\FF_{p^2}$ can achieve the capacity of this FF-BC.
Intuitively, this is because there is no additional power cost with Zero-Forcing Beamforming (ZFB) in the finite-field domain
(unlike ZFB in the complex domain).

\begin{theorem}\label{lem:sumRQCoF}
Consider the FF-BC in (\ref{eq:FF-BC}) for $K = L$. If $\tilde{\Qm}$ has rank $L$, the sum capacity is
\begin{equation} \label{straminchia}
C_{\mbox{\tiny{FF-BC}}} = 2L\log{p} - \sum_{\ell=1}^{L} H(\zeta(\tilde{\hv}_{\ell},\tilde{\av}_{\ell})).
\end{equation}
and it can be achieved by linear coding.
\end{theorem}
\begin{IEEEproof} See Appendix \ref{proof:FF-BC}.
\end{IEEEproof}
Motivated by Theorem \ref{lem:sumRQCoF}, we present the RQCoF scheme
using finite-field matrix inversion precoding at the CP. As for RCoF, we use $L$ nested linear codes
$\Cc_{L} \subseteq \cdots \subseteq \Cc_{1}$ where $\Cc_{\ell}$ has rate $R_{\ell}=\frac{2r_{\ell}}{n}\log{p}$ and let
$k_\ell$ denote the UT destination of the $\ell$-th message, encoded by $\Cc_\ell$.
The CP precodes the zero-padded information messages  $\{\underline{\tilde{\wv}}_{\ell}:\ell=1,\ldots,L\}$
as in (\ref{eq:precod}) and sends the precoded message $\underline{\tilde{\muv}}_{\ell}$  to  AT $\ell$ for all $\ell = 1,\ldots, L$, via the digital
backhaul link. AT $\ell$ generates the codeword $\underline{\tilde{\cv}}_\ell = \underline{\tilde{\muv}}_{\ell} \Gm \in \Cc_1$, and the corresponding
transmitted signal $\underline{\tilde{\xv}}_\ell$ according to (\ref{transmit-quant}), with $\underline{\tilde{\tv}}_\ell = f(\underline{\tilde{\muv}}_{\ell})$
Each UT $k_\ell$ produces its quantized output according to the scalar mapping (\ref{eq:demo}) and obtains:
\begin{eqnarray}
\underline{\tilde{\uv}}_{k_\ell} & = &
\left ( \tilde{\qv}_{k_\ell}^\transp
\left[
                \begin{array}{c}
                  \underline{\tilde{\cv}}_{1} \\
                  \vdots \\
                  \underline{\tilde{\cv}}_{L} \\
                \end{array}
              \right] \right ) \oplus  \underline{\zetav}(\tilde{\hv}_{k_\ell},\tilde{\av}_{k_\ell}) \nonumber \\
& = & \left ( \tilde{\qv}_{k_\ell}^\transp
\tilde{\Qm}^{-1} \left[
                \begin{array}{c}
                  \underline{\tilde{\wv}}_1\Gm \\
                  \vdots \\
                  \underline{\tilde{\wv}}_L \Gm \\
                \end{array}
              \right] \right ) \oplus  \underline{\zetav}(\tilde{\hv}_{k_\ell},\tilde{\av}_{k_\ell}) \nonumber \\
& = &
               \underline{\tilde{\vv}}_\ell \oplus  \underline{\zetav}(\tilde{\hv}_{k_\ell},\tilde{\av}_{k_\ell}) \label{precoded-ffbc}
\end{eqnarray}
where $\underline{\tilde{\vv}}_\ell = \underline{\tilde{\wv}}_\ell \Gm$ is a codeword of $\Cc_\ell$.
Thus, UT $k_\ell$ can recover its desired message if $R_{\ell} \leq 2 \log{p} -  H(\zeta(\tilde{\hv}_{k_\ell},\tilde{\av}_{k_\ell}))$. Summarizing, we have:

\begin{theorem}\label{thm:eRQCoF}
RQCoF applied to a DAS downlink  with channel matrix $\tilde{\Hm}=[\tilde{\hv}_{1},\ldots,\tilde{\hv}_{L}]^{\transp}\in \CC^{L \times L}$
achieves the sum rate
\begin{eqnarray}
R_{\mbox{\tiny{RQCoF}}}(\tilde{\Hm},\tilde{\Am}) &=& \sum_{\ell=1}^{L} \min\left \{R_{0},  2 \log{p} - H(\zeta(\tilde{\hv_{\ell}},\tilde{\av_{\ell}})) \right \}.
\end{eqnarray}
\hfill \IEEEQED
\end{theorem}

\section{Comparison with Known Schemes on the Wyner Model} \label{sec:Wyner}

In order to obtain clean performance comparisons with other state-of-the art information theoretic coding strategies,
we consider the symmetric Wyner model \cite{Wyner},
which has been used in several other works for its simplicity and analytic tractability.
In particular, we consider comparisons with Quantize reMap and Forward (QMF) and Decode and Forward (DF) for
the DAS uplink, and Compressed Dirty Paper Coding (CDPC) and Compressed Zero-Forcing Beamforming (CZFB)
for the DAS downlink.

In the symmetric Wyner model with $L$ ATs and $L$ UTs, the received signal at the $\ell$-th receiver (AT for the uplink or UT for the downlink)  is given by
\begin{equation}
\underline{\yv}_{\ell} = \underline{\xv}_{\ell} + \gamma(\underline{\xv}_{\ell-1} + \underline{\xv}_{\ell+1}) + \underline{\zv}_{\ell}
\end{equation}
where $\gamma \in (0,1]$ quantifies the strength of inter-cell interference and $\underline{\zv}_{\ell}$ has i.i.d. components
$\sim \Cc\Nc(0,1)$.

\subsection{Review of some Classical Coding Strategies}  \label{sec:coding-strategies}

\subsubsection{QMF}\label{subsec:CF}
Each AT performs vector quantization of its received signal at some rate
$R' \geq R_{0}$ and maps the blocks of $nR'$ quantization bits into binary words of length $nR_{0}$  by using a hashing function (binning).
The CP performs the joint decoding of all UTs' messages based on the observation of all the (hashed) quantization bits.
Using random coding with Gaussian codes and random binning, \cite{Sanderovich} proves the following achievable rate of QMF:
\begin{eqnarray}
R_{\tiny{\mbox{QMF}}} &=&  \max_{0 \leq r} \min_{\Sc \subset [1:L]} \Big\{|\Sc|(R_{0}-r)+\log\det\left (\Id+\SNR(1-2^{-r})\Hm(\Sc^{c},[1:L])\Hm(\Sc^{c},[1:L])^{\herm} \right ) \Big\}. \nonumber \\
& &  \label{eq:sumrateQMF}
\end{eqnarray}
As $R_{0} \rightarrow \infty$,  $R_{\tiny{\mbox{QMF}}}$ tends to the sum rate of the underlying multi-antenna G-MAC channel with $L$ users and one $L$-antennas receiver. For  $\SNR \rightarrow \infty$ and fixed $R_0$, then $R_{\tiny{\mbox{QMF}}} \rightarrow L R_{0}$ \cite{Sanderovich,Nazer2009}.
While for a general channel matrix computing (\ref{eq:sumrateQMF}) is non-trivial, a remarkable result of \cite{Sanderovich} is that
for the Wyner model in the limit of $L \rightarrow \infty$ the QMF rate per user can be simplified to
\begin{equation}  \label{eq:QMFsimple}
R_{\tiny{\mbox{QMF, per-user}}} = F(r^{*}),
\end{equation}
where
\begin{equation*}
F(r) = \int_{0}^{1}\log \left (1+\SNR\left (1-2^{-r} \right )\left (1+2 \gamma \cos (2\pi\theta) \right )^{2} \right )d\theta.
\end{equation*}
and where $r^*$ is the solution of the equation $F(r) = R_{0} - r$.
A simplified version of QMF does not use binning, and simply forwards to the CP the quantization bits collected at the ATs. We refer to this scheme as
{\em Quantize and Forward} (QF), without the re-mapping. In this case, the quantization rate is $R' = R_0$.
From \cite{sanderovich2009},  the achievable sum rate of QF is given by
\begin{eqnarray}
R_{\tiny{\mbox{QF}}} &=& \log\det \left (\Id + \SNR \; \Dm\Hm\Hm^{\herm} \right ),\label{eq:sumrateQMFs}
\end{eqnarray}
where $\Dm = \diag(1/(1+D_{\ell}) : \ell = 1,\ldots, L)$ and $D_{\ell} = (1+\SNR \| \hv_\ell\|^2)/(2^{R_{0}}-1)$ denotes the variance
of quantization noise at AT $\ell$.

\subsubsection{DF}\label{subsec:DF}
In the Wyner model, each AT $\ell$ sees the three-inputs G-MAC formed by UTs $\ell-1$, $\ell$ and $\ell+1$. Imposing either to treat interference as noise, or
to decode all messages at each AT, yields \cite{Sanderovich}:
\begin{eqnarray*}
R_{1} &=& \log\Big(1+\frac{\SNR}{1+2\gamma^2 \SNR}\Big)\\
R_{2} &=& \min\Big\{\frac{1}{2}\log(1+2 \gamma^2 \SNR), \frac{1}{3}\log(1+(1+2\gamma^2)\SNR) \Big\}\\
R_{\tiny{\mbox{sum}}} &=& L \times \min\{\max(R_{1},R_{2}),R_{0}\}.
\end{eqnarray*}
This scheme has no joint-processing gain. However, when $R_{0}$ is sufficiently small compared to the rates achievable over the
wireless channel, or when $\gamma$ is very small, this scheme can be optimal \cite{Sanderovich,Nazer2009}.
In fact, DF is what is implemented today in a network of small cells, where each AT operates as a stand-alone base station,
and the decoded packets are sent to a common node that may use packet selection macro-diversity, in the case some of the base stations fail
to decode. Therefore, it is useful to compare with DF to quantify the potential gains of other schemes with respect to current technology.

\subsubsection{CDPC}\label{subsec:CDPC}
We focus now on the DAS downlink.  In CDPC the CP performs joint DPC under per-antenna power
constraint and sends the compressed (or quantized) DPC codewords to the corresponding
ATs via wired links. As a consequence, the ATs also transmit quantization noise.
Let $\underline{\tilde{\vv}}_{\ell}$ be the DPC-encoded signal to be transmitted by AT $\ell$
and let $\dot{\underline{\vv}}_{\ell}$ denote its quantized version. Define
$\sigma_\ell^2 = \frac{1}{n} \EE[\|\underline{\tilde{\vv}}_\ell\|^2]$ and
$\dot{\sigma}_\ell^2 = \frac{1}{n} \EE[\|\underline{\dot{\vv}}_\ell\|^2]$.
From the standard rate distortion theory, an achievable quantization distortion $D_\ell$ is given by
\begin{eqnarray*}
R(D_{\ell}) & \eqdef & \min_{P_{\hat{V}_{\ell} | V_{\ell}}: \EE[\|V_{\ell}-\hat{V}_{\ell}\|^2] \leq D_{\ell}} \; I(V_{\ell};\hat{V_{\ell}}) \\
& \leq &  I(V_{\ell};\dot{V}_{\ell}) = \log(1 + \sigma_\ell^2/D_{\ell}),
\end{eqnarray*}
where the upper bound follows from the choice
$\dot{V}_{\ell} = V_{\ell} + \dot{Z}_{\ell}$ with $\dot{Z}_{\ell} \sim \Cc\Nc(0,D_{\ell})$ and $V_\ell$ with variance $\sigma_\ell^2$.
Letting $R_{0} = \log(1+\sigma_\ell^2 / D_{\ell})$ and solving for $D_\ell$ we obtain
\begin{equation}\label{eq:quant}
D_{\ell} = \frac{\sigma_\ell^2}{2^{R_{0}}-1}.
\end{equation}
Using the fact that $\dot{\sigma}_\ell^2 = \sigma_\ell^2 + D_\ell$,
the per-antenna power constraint $\dot{\sigma}_\ell^2 \leq \SNR$ imposed at each AT $\ell$ yields
\begin{equation}
\sigma_\ell^2 \leq \SNR \frac{2^{R_{0}}-1}{2^{R_{0}}} \;\;\; \mbox{for} \; \ell = 1,\ldots, L. \label{eq:DPCpower}
\end{equation}
Using (\ref{eq:DPCpower}) in (\ref{eq:quant}), we obtain $D_{\ell} = \SNR \; 2^{- R_{0}}$ for $\ell = 1,...,L$.
At the $\ell$-th UT receiver, the variance of the effective noise is given by
\begin{equation}\label{eq:DPCvar}
\tilde{\sigma}_{\ell}^{2} = 1 + \| \tilde{\hv}_{\ell} \|^2 \SNR \; 2^{- R_{0}}.
\end{equation}
Then, an achievable sum rate of CDPC is equal to the sum capacity of the resulting vector BC
with the above modifications (i.e., per-antenna power constraint and noise variance).
This can be computed using the efficient algorithm given in \cite{Hoon}, based on Lagrangian duality.
Further, the closed form rate-expression was provided in \cite{Simeone} for the so-called {\em soft-handoff} Wyner model,
a simplified variant  of the Wyner model where each receiver has only one interfering signal from its left  neighboring cell.
While CDPC is expected to be near optimal for large $R_{0}$, it is generally suboptimal at finite (possibly small) $R_{0}$.

\subsubsection{CZFB}\label{subsec:CZFB}
CP performs precoding with the inverse channel matrix $\Bm = \tilde{\Hm}^{-1}$ and sends the compressed ZFB signals to the corresponding ATs via wired links.
As in CDPC, the ATs forward also quantization noise, such that the variance of effective noise at the $\ell$-th UT is given again by (\ref{eq:DPCvar}).
The transmit power constraint (\ref{eq:DPCpower}) holds verbatim. Because of the non-unitary precoding, the {\em useful} signal power
is given by $\SNR \frac{2^{R_{0}}-1}{2^{R_{0}} \|\bv_\ell\|^2}$ where $\bv_\ell^\transp$ is the $\ell$-th row of the precoding matrix $\Bm$.
It follows that CZFB achieves the sum rate
\begin{eqnarray}
R_{\mbox{\tiny CZFB}} &=& \sum_{\ell=1}^{L}\log\left(1+\frac{\SNR / \|\bv_{\ell}\|^2}{1+(1+\|\tilde{\hv}\|^2\SNR)/(2^{R_{0}}-1)}\right).
\end{eqnarray}

\subsection{Numerical Results}

\begin{figure}
\centerline{\includegraphics[width=11cm]{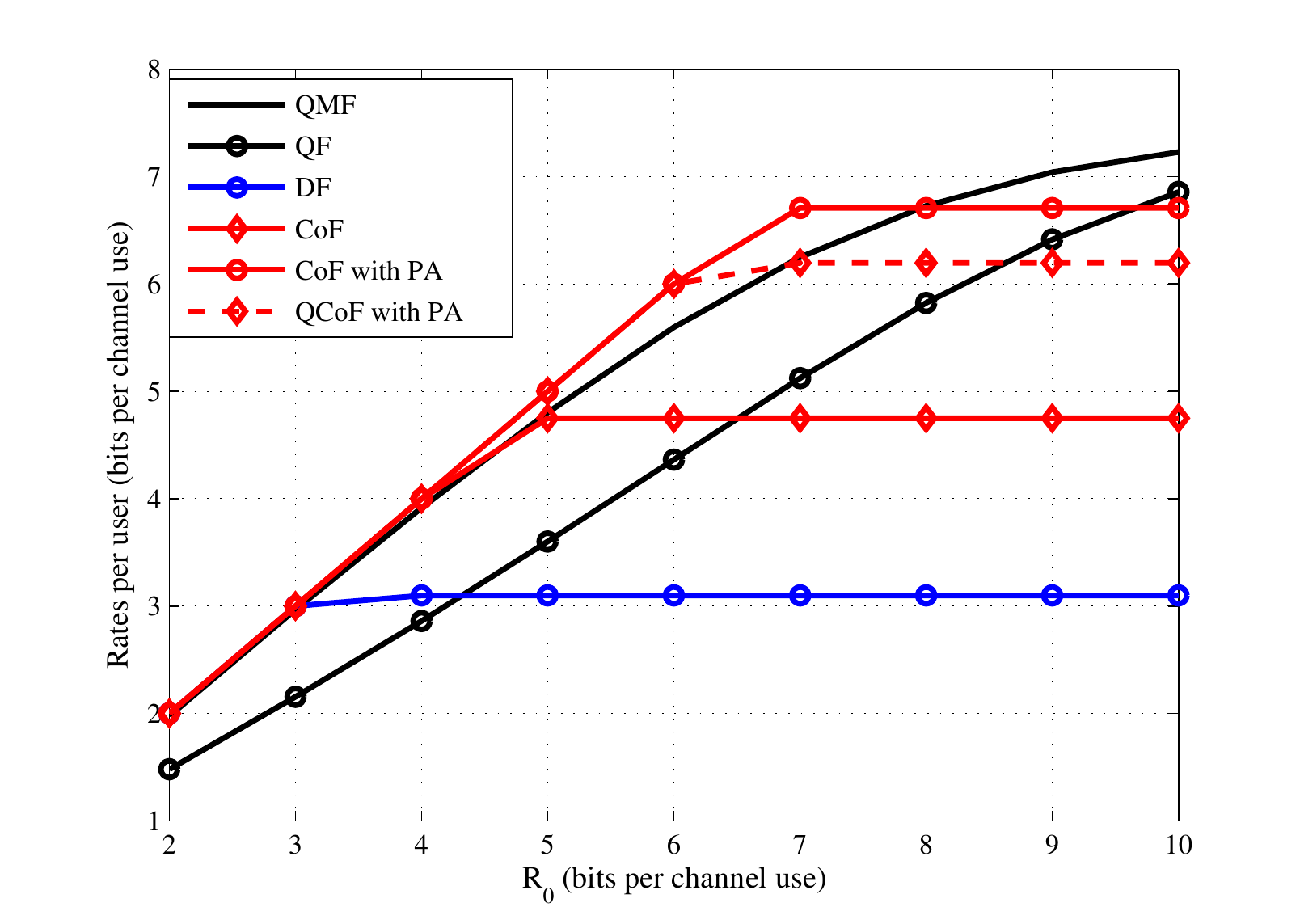}}
\caption{$\SNR = 25$dB and $L=\infty$. Achievable rates per user as a function of $R_{0}$, for the DAS uplink in the Wyner model case with inter-cell
interference parameter $\gamma=0.7$.}
\label{simulation_ULWyner}
\end{figure}

\begin{figure}
\centerline{\includegraphics[width=11cm]{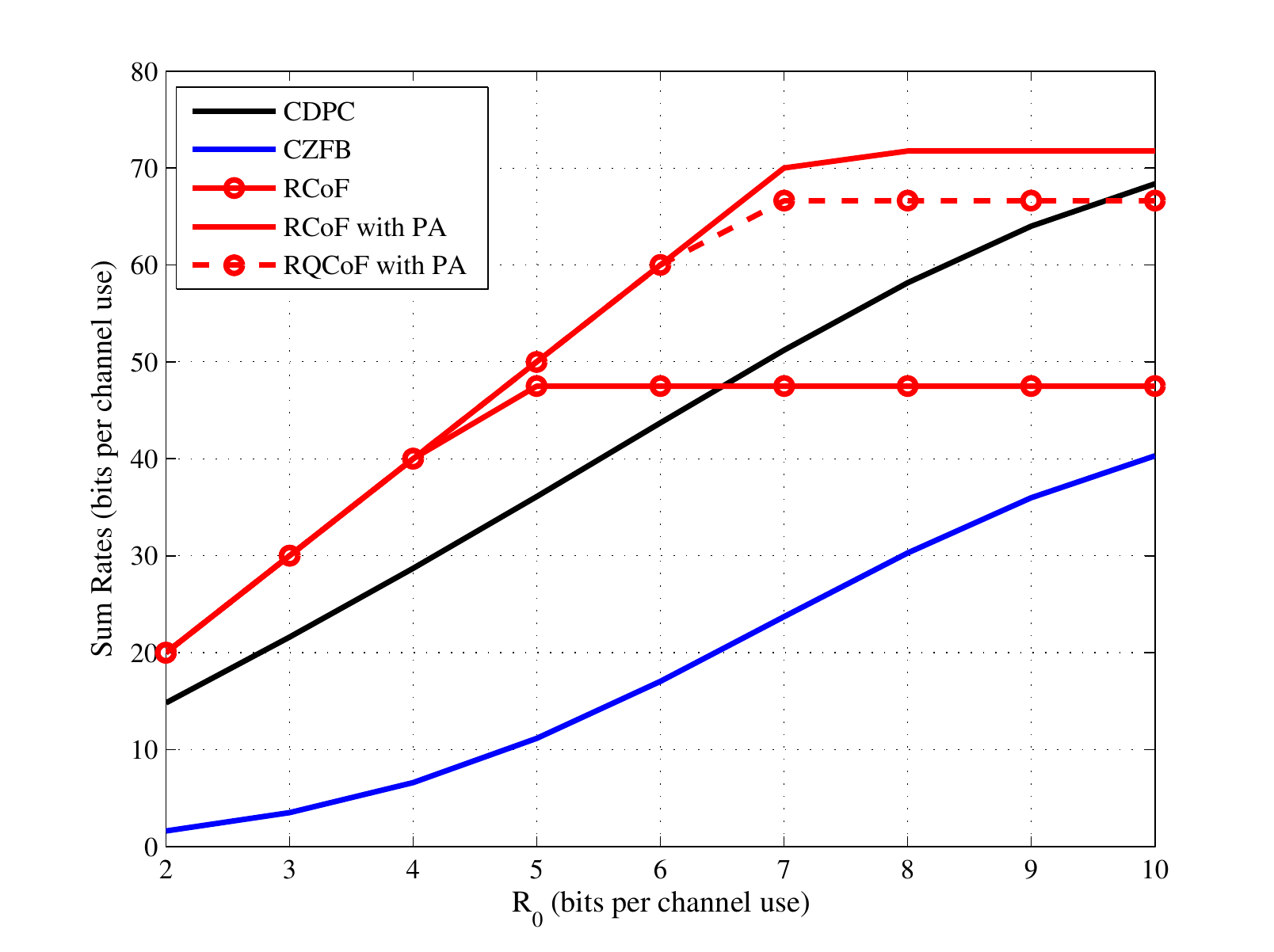}}
\caption{$\SNR = 25$dB and $L=10$. Achievable sum rates as a function of $R_{0}$, for the DAS downlink in the  Wyner model case with
inter-cell interference parameter $\gamma=0.7$.}
\label{simulation_DLWyner}
\end{figure}

Thanks to the banded structure of the Wyner model channel matrix,
the resulting system matrix of CoF (resp.,  RCoF)
is guaranteed to have rank $L$ although every AT (reps., UT) determines its integer coefficients vector in a distributed way.
In addition, the non-integer penalty which may be relevant for specific values of $\gamma$ can be mitigated by using a
{\em power allocation} strategy, in order to create more favorable channel coefficients for the integer conversion at each receivers.
In \cite{Nazer2009} a further improved strategy is proposed based on superposition coding, where the user messages are split into
two layers, and one layer is treated as noise while the other is treated by CoF. Here we focus on simple power allocation, since it is simpler,
practical, and captures a significant fraction of the gains achieved with superposition coding. The power allocation strategy
works as follows: odd-numbered UTs (resp.,  ATs) transmit at power $\beta P$ and even-numbered UTs (reps., ATs) transmit at power $(2-\beta)P$,
for $\beta \in [0,1]$. The role of odd- and even-numbered UTs (or ATs) is alternately reversed in successive time slots,
such that each UT (resp., AT) satisfies its individual power constraint on average.
Accordingly, the effective coefficients of the channel for odd-numbered and even-numbered relays are
$\hv_{o} = [\gamma \sqrt{2-\beta}, \sqrt{\beta}, \gamma\sqrt{2-\beta}]$
and $\hv_{e}=[\gamma \sqrt{\beta}, \sqrt{2-\beta}, \gamma\sqrt{\beta}]$. For given $\gamma$, the parameter $\beta \in [0,1]$ can be optimized
to make the effective channels better suited for the integer approximation in the CoF receiver mapping.
We have two computation rates, $R(\hv_{o},\av_{o})$ and $R(\hv_{e},\av_{e})$, at the odd and even numbered receivers.
The achievable symmetric rate of CoF (or RCoF) with power allocation is given by $\min\{R_{0},R(\hv_{o},\av_{o},\SNR),R(\hv_{e},\av_{e},\SNR)\}$.
Notice that the odd- and even-numbered relays can optimize their own equation coefficients independently,
but the optimization with respect to $\beta$ is common to all, and the computation rate is the minimum computation rate over all the relays,
since the same lattice code $\Lc$ is used across all users.
In Fig.~\ref{simulation_ULWyner}, we show the performance of various relaying strategies for the DAS uplink with $\SNR=25$ dB, as a function of
backhaul rate $R_{0}$. $L = \infty$ is assumed in order to use the simple rate expression of QMF in (\ref{eq:QMFsimple}).
Fig.~\ref{simulation_ULWyner} shows that the power allocation strategy significantly reduces the integer approximation penalty and almost
achieves the cut-set bound outer bound (i.e., capacity) for $R_{0} \leq 7$ bits.
Not surprisingly,  QCoF with $p=251$ only pays the shaping penalty with respect to CoF,
i.e., it approaches the performance of the corresponding high-dimensional scheme within $\approx 0.5$ bit per complex dimension.

We observe a similar trend for the downlink schemes, shown in Fig.~\ref{simulation_DLWyner}.
In this case, the achievable sum rate of RCoF with power allocation is given by
\begin{equation}
R_{\mbox{\tiny sum}} = \frac{1}{2}(\min\{R_{0},R(\hv_{o},\av_{o},\SNR)\} + \min\{R_{0},R(\hv_{e},\av_{e},\SNR)\}),
\end{equation}
where the average, sainted of the minimum, between odd and even numbered UTs is due to the fact that in RCoF
we can use two different lattice codes and therefore the rates aren't constrained to be all equal.
RCoF outperforms CDPC for $R_{0} \leq 6.5$ bits per channel use.

It is remarkable to observe that  the fully practical and easily implementable quantized schemes
QCoF and RQCoF can outperform other conventional practical schemes such as DF and CZFB, respectively.
These results show that CoF and RCoF are good candidate for DAS uplink and downlink, respectively, in particular
in the regime of small to moderate $R_{0}$ and high SNR.
This regime is relevant for small cell networks with limited backhaul cooperation,
where the backhaul becomes the system bottleneck.
Further, we observed that the proposed schemes can be significantly improved by mitigating the impact of the non-integer penalty.
In this model, power allocation is effective due to the system symmetric structure.
However, it is not clear how to extend the power allocation approach in the general case of a wireless network
whose channel matrix is the result of fading, shadowing and pathloss, and therefore it does not enjoy any special easily parameterized structure.
In the next section we address the case of a general wireless network with random channel coefficients, and show that
{\em multiuser diversity} (i.e., AT/UT selection) can greatly improve the performance of the
basic schemes.

\section{Antenna and User Selection}\label{sec:sch}

Since the proposed schemes require an equal number of ATs and UTs active at each given time, in a general
DAS with $K$ UTs and $L$ ATs the system must select which terminals are active in every scheduling slot.
We define the ``active" set of UTs $\Uc \subseteq [1:K]$ as the subset of UTs that are actually scheduled for transmission (resp., reception)
on the current uplink (resp., downlink) slot, comprising $n$ channel uses.  Similarly, the ``active" set of ATs $\Ac \subseteq [1:L]$ is defined as the
subset of ATs that are used for reception (resp., transmission) on the current uplink (resp., downlink) slot.

\subsection{Antenna Selection for the DAS Uplink}

We assume that the active set of UTs is fixed a priori. Without loss of generality, we can fix $\Uc = [1:K]$ and assume $K  < L$. Our goal is to select
a subset $\Ac \subset [1:L]$ of ATs of cardinality $K$.  Recall that every AT chooses the integer combination coefficients, and therefore its vector $\qv_{\ell}$,
using Algorithm 1 in order to maximize its own computation rate $R_\ell = R(\hv_\ell, \av_\ell,\SNR)$. The CP knows $\{\qv_\ell, R_\ell : \ell \in [1:L]\}$.
The CP aims at maximizing the sum rate such that the resulting system matrix is full-rank, by selecting a subset of ATs for the given UT active set $\Uc$.

\subsubsection{AT selection for CoF (or QCoF)}\label{subsubsec:schCoF}

From Theorem \ref{thm:CoF}, the AT selection problem consists of finding $\Ac$ solution of:
\begin{eqnarray}
\max_{\Ac \subset [1:L]} && \sum_{s=1}^{S(\Ac)} |\Ac_{s}| \min\{R_{0},\min \{R_{\ell}: \ell \in \Ac_{s}\}\}\label{opt:obj}\\
\mbox{subject to} &&\mbox{Rank}(\Qm(\Ac,\Uc)) = |\Uc|, \label{opt:const}
\end{eqnarray}
where $S(\Ac)$ indicates the number of disjoint subnetworks with respect to $\Qm(\Ac,\Uc)$.
This problem has no particularly nice structure and the optimal solution is obtained, in general, by exhaustive search
over all $|\Uc| \times |\Uc|$ submatrices of $\Qm([1:L],\Uc)$.
Yet, we notice that if an optimal solution $\Ac^{\star}$ does not decompose
(i.e., $S(\Ac^{\star})=1$), the simple greedy Algorithm 2 given below
finds it (see Lemma \ref{lem:greedy}).
Namely, there exists a low-complexity algorithm to find an optimal AT selection for dense networks whose system matrix $\Qm([1:L],\Uc)$ cannot be decomposed
in block-diagonal form.

In general, we may have several disjoint subnetworks, each of which does not decompose further, even when removing
some ATs.  Then, we can perform antenna selection by using Algorithm 2 on each subnetwork component.
If the optimum solution of each subnetwork component does not involve further network decomposition,
by Lemma \ref{lem:greedy} we are guaranteed to arrive at an optimal global solution.
This generally suboptimal (but efficient) approach can be summarized as
\begin{itemize}
\item For given $\Qm = \Qm([1:L],\Uc)$, perform network decomposition using
depth-first or breadth-first search \cite{Ahuja}, yielding disjoint subnetworks $\Qm(\Ac_{s},\Uc_{s})$ for $s=1,\ldots,S$.
\item For each subnetwork $\Qm(\Ac_{s},\Uc_{s})$, run Algorithm $2$ and find a good selection $\Ac_{s}^{\star} \subset \Ac_{s}$ with
$|\Ac_{s}^{\star}|=|\Uc_{s}|$.
\item Finally, obtain the set of active ATs, $\Ac^{\star} = \cup_{s=1}^{S} \Ac_{s}^{\star}$, such that $|\Ac^{\star}| = |\Uc|$.
\end{itemize}

\begin{algorithm} \label{alg:greedy}
\caption{The Greedy Algorithm}
\textbf{Input}: ($\Qm$, $\{w_{\ell}: \ell=1,\ldots,m\})$ where $\Qm$ is a full-rank $m \times n$ matrix with $m > n$

\textbf{Output}: $\Sc \subset [1:m]$ with $|\Sc| = n$

\begin{enumerate}
\item Sort $[1:m]$ such that $w_{1} \geq w_{2} \geq \cdots \geq w_{m}$
\item Initially, $\ell=1$ and $\Sc = \emptyset$
\item If $\hbox{Rank}(\Qm(\Sc \cup \{\ell\},[1:n])) > \hbox{Rank}(\Qm(\Sc,[1:n]))$,\\
then $\Sc \leftarrow \Sc \cup \{\ell\}$
\item  Set $\ell = \ell+1$
\item  Repeat 3)-4) until $|\Sc| = n$
\end{enumerate}
\end{algorithm}

We have

\begin{lemma}\label{lem:greedy}
If $\mbox{Rank}(\Qm) = n$, Algorithm $2$ finds a solution to the problem
\begin{eqnarray}
\max_{\Sc \subset [1:m]} && \min \{w_{\ell} : \ell \in \Sc\} \label{opt:CoF1}\\
\mbox{subject to}&& \mbox{Rank}(\Qm(\Sc,[1:n])) = n. \label{opt:CoF2}
\end{eqnarray}
\end{lemma}
\begin{IEEEproof}
Let $\hat{\Qm}$ be the row-permuted matrix of $\Qm$ according to the decreasing ordering of the weights  $w_{\ell}$.
The problem is then reduced to finding the minimum row index $\ell^{\dagger}$
such that $\hat{\Qm}([1:\ell^{\dagger}],[1:n])$ has rank $n$. This is precisely what  Algorithm 2 does.
\end{IEEEproof}

An immediate corollary of Lemma \ref{lem:greedy} is that, if one disregards network decomposition,
then Algorithm 2 finds the maximum computation rate over the AT selection. In fact, it is sufficient to use Algorithm 2 with
$m = L$, $n = K$, and input $\Qm = \Qm([1:L],\Uc)$ and $w_\ell = \min\{R_0,R_\ell\}$ for $\ell = 1,\ldots, L$.

\subsubsection{AT selection for LQF}\label{subsubsec:schLQF}

From Theorem \ref{lem:LQF}, the AT selection problem consists of finding $\Ac$ solution of:
\begin{eqnarray}
\max_{\Ac \subset [1:L]} && \sum_{\ell \in \Ac} \min\{R_0,R_{\ell}\} \label{opt:objLQF}\\
\mbox{subject to} &&\mbox{Rank}(\Qm(\Ac,\Uc)) = |\Uc|,\label{opt:constLQF}
\end{eqnarray}
where we let $R_\ell = 2\log p - H(\zeta(\hv_\ell,\av_\ell))$ (see Section \ref{subsec:QCoF}).
This problem consists of the maximization of linear function subject to a matroid constraint,
where the matroid $\Mc = (\Omega,\Ic)$ is defined by the ground set $\Omega = [1:L]$ and by the collection of independent sets
$\Ic = \{\Ac \subseteq \Omega: \Qm(\Ac,\Uc) \mbox{ has linearly independent rows}\}$.
Rado and Edmonds \cite{Rado, Edmonds} proved that a greedy algorithm finds an optimal solution. In this case, such algorithm coincides with
Algorithm $2$ with input $ \Qm = \Qm([1:L],\Uc)$ and  $w_\ell = \min\{R_0,R_{\ell}\}$.

\subsection{User Selection for the DAS Downlink}

In this case we assume that the set of ATs $\Ac = [1:L]$ is fixed and $K > L$. Hence, we wish to select a subset $\Uc \subset [1:K]$ of cardinality $L$
such that the resulting system matrix has rank $L$ and the DAS downlink sum rate is maximized.
The CP has knowledge of the downlink system matrix $\tilde{\Qm}([1:K],\Ac) = [\tilde{\qv}_{1},\ldots,\tilde{\qv}_{K}]^{\transp}$ and the set of individual
user computation rates, $\tilde{R}_k = R(\tilde{\hv}_k,\tilde{\av}_k,\SNR)$ for RCoF, or
$\tilde{R}_k =  2\log p - H(\zeta(\tilde{\hv}_k,\tilde{\av}_k))$ for RQCoF (see Theorem \ref{lem:sumRQCoF}).
The UT selection problem consists of finding $\Uc$ solution of:
\begin{eqnarray}
\max_{\Uc \subset [1:K]} && \sum_{k \in \Uc} \min\{R_{0}, \tilde{R}_{k}\}\\
\mbox{subject to} &&\mbox{Rank}(\tilde{\Qm}(\Uc,\Ac)) = |\Ac|.
\end{eqnarray}
As noticed before, this can be regarded as the maximization of linear function over matroid constraint.
Therefore, Algorithm 2  with input $\Qm = \tilde{\Qm}([1:K],\Ac)$ and $w_k = \min\{R_{0},\tilde{R}_{k}\}$ provides an optimal solution.

\subsection{Comparison on the Bernoulli-Gaussian Model}\label{sec:BG}

\begin{figure}
\centerline{\includegraphics[width=11cm]{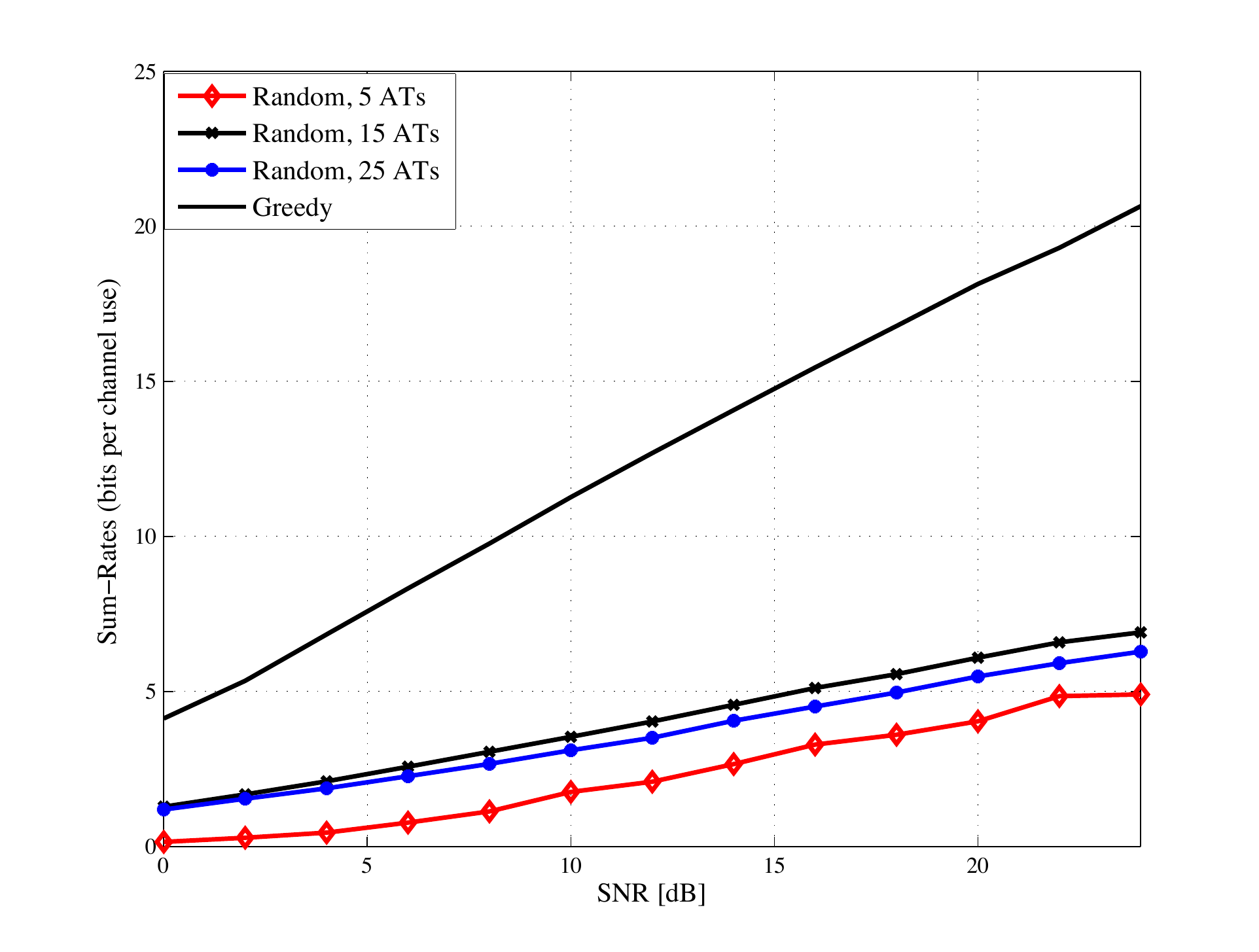}}
\caption{DAS uplink with $K = 5$, $L = 25$ and $R_0 = 6$ bit/channel use: average sum rate vs. SNR
on the Bernoulli-Gaussian model with $q = 0.5$.}
\label{simulation_UL}
\end{figure}

\begin{figure}
\centerline{\includegraphics[width=11cm]{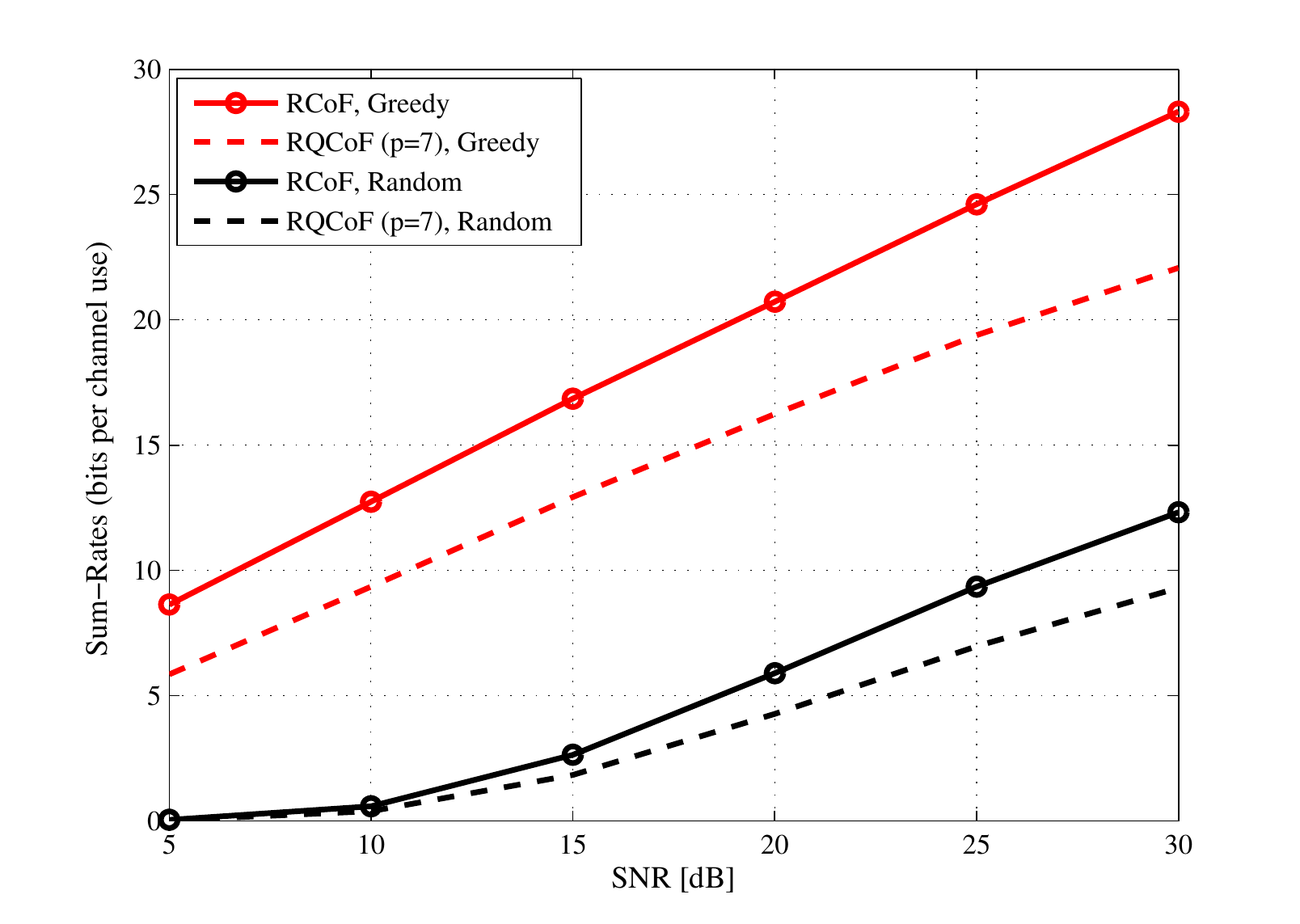}}
\caption{DAS downlink with $K = 25$, $L = 5$ and $R_0 = 6$ bit/channel use: average sum rate vs. SNR
on the Bernoulli-Gaussian model with $q = 0.5$.}
\label{simulation_DLsch}
\end{figure}

We consider a DAS with channel matrix with i.i.d. elements $[\Hm]_{\ell,k} = h_{\ell,k} \gamma_{\ell,k}$, where
$h_{\ell,k} \sim \Cc\Nc(0,1)$ and $\gamma_{\ell,k}$ is a Bernoulli random variable with $\PP(\gamma_{\ell,k} = 1) = q$.
This model captures the presence of Rayleigh fading and some extreme form of path-blocking
shadowing, and it is appropriate for a DAS deployed in buildings, or dense urban environments where the ATs
are not mounted on tall towers, in contrast to conventional macro-cellular systems. For the downlink results, we assume a channel matrix $\tilde{\Hm}$ with the same statistics.

We compute the {\em ergodic} sum rates by Monte Carlo averaging with respect to the
channel matrix.  If the resulting system matrix, after AT (resp., UT) selection is rank deficient,
then the achieved instantaneous sum rate is zero, for that specific realization. Hence, rank deficiency can be
regarded as a sort of ``information outage'' event.
With the path gain coefficients and noise variance normalization adopted here,
the SNR coincides with the individual nodes power constraint.

Fig.~\ref{simulation_UL} shows the average sum rate for a DAS uplink with $K = 5$ UTs, $L = 25$ ATs and channel blocking probability
$q = 0.5$. This result clearly show that the proposed ``greedy'' AT selection scheme yields a large improvement over random selection of a fixed number of ATs, and essentially eliminates the problem of system matrix rank deficiency, provided that $L \gg K$.
The curves denoted as ``random selection'' indicate the case where a fixed number $L' < L$ of ATs  is randomly and uniformly selected,
independent of channel realizations. For $L'=25$ the DAS uses all the available ATs all the time, yet its performance is much worse than selecting $5$ ATs out of $25$ according to the proposed selection scheme. Fig.~\ref{simulation_DLsch} shows a similar trend for the DAS downlink.
Here, random selection indicates that $5$ UTs are chosen at random out of the $25$ UTs.
We notice that the sum rate vs. SNR curves for both greedy and random UT selection have the same slope, indicating that the rank-deficiency problem
is not significant in both cases. However, greedy selection achieves a very evident {\em multiuser diversity} gain over random selection. This is not
only due to selecting channel vectors with large gains, as in conventional multiuser diversity, but also to the fact that the greedy selection is able to
choose channels that are adapted to the RCoF strategy, i.e., whose coefficients are well approximated by integers (up to a common scaling factor).
It is also interesting to notice that RQCoF with greedy selection does not suffer from the rank-deficiency of the system matrix even for
$p$ as small as 7, in the example. This is indicated by the fact that the sum rate gap between RQCoF and RCoF is essentially equal to the
shaping loss (0.5 bits per user).

\begin{figure}
\centerline{\includegraphics[width=11cm]{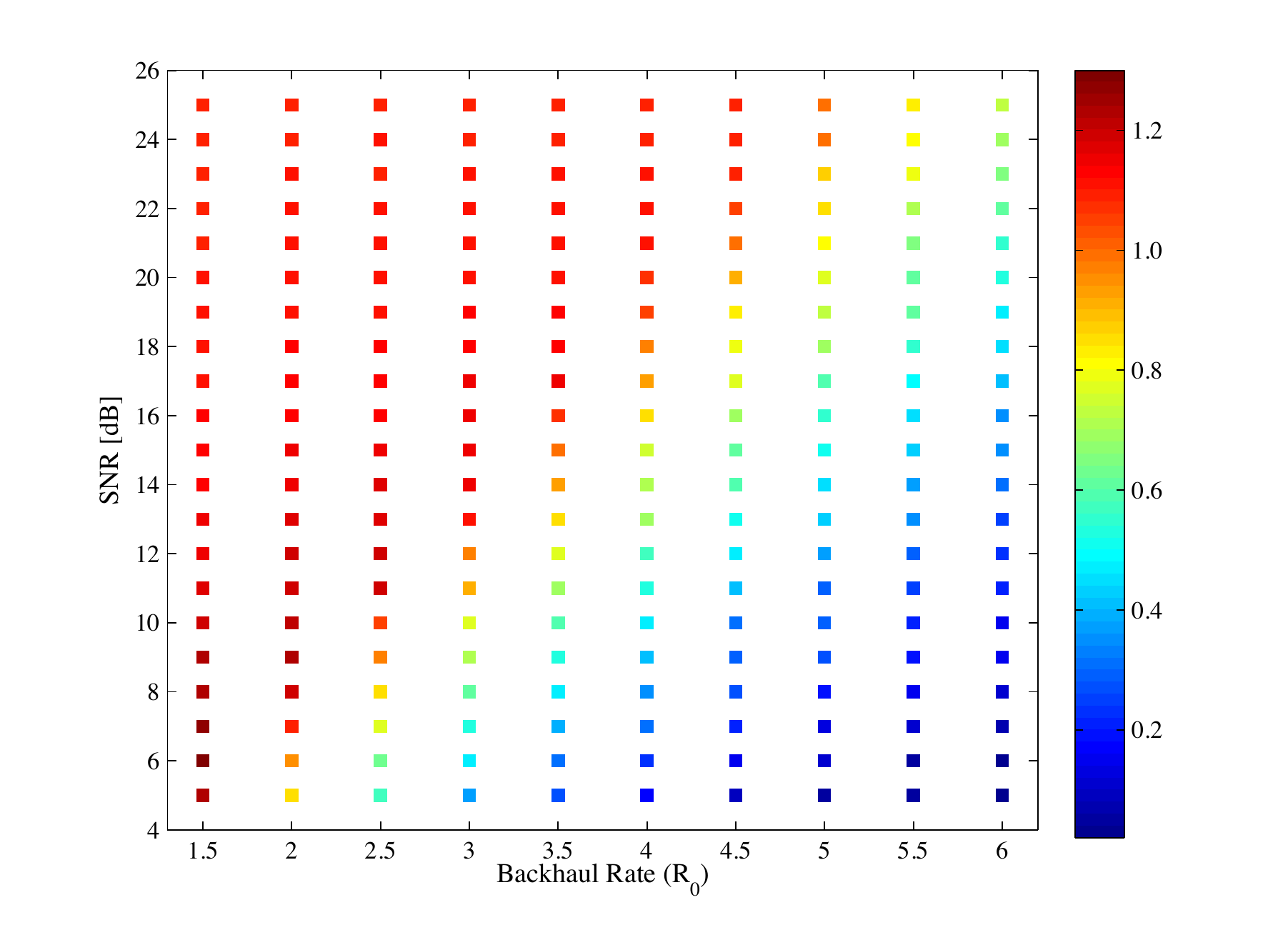}}
\caption{DAS uplink with $K = 5$ and $L = 50$,
Bernoulli-Gaussian model with $q = 0.5$:  Colors represent the relative gain of CoF versus QF (e.g., ratio of sum rates $R_{\mbox{\tiny{CoF}}}/R_{\mbox{\tiny{QF}}}$).}
\label{simulation_uplinkcol}
\end{figure}

\begin{figure}
\centerline{\includegraphics[width=11cm]{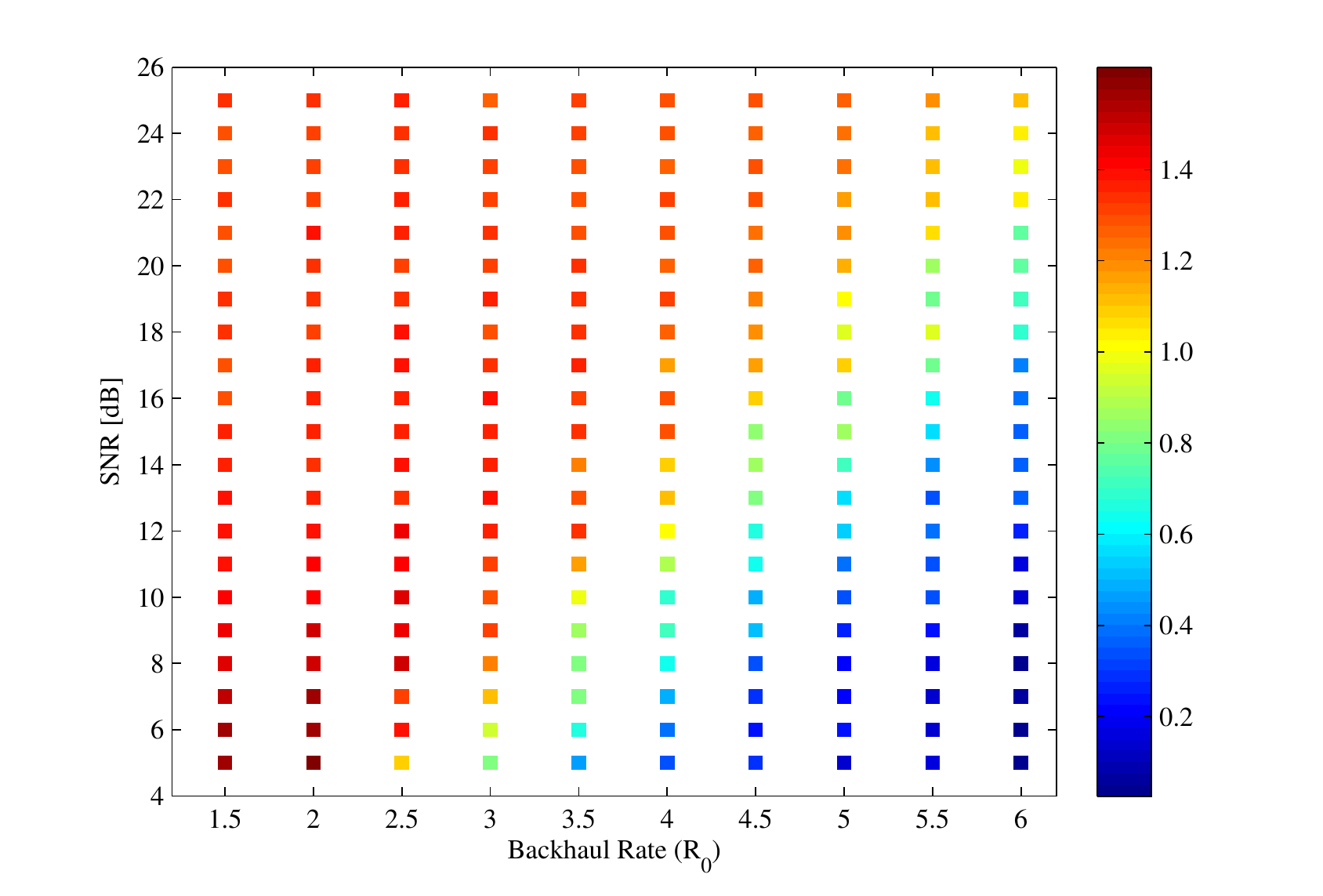}}
\caption{DAS downlink with $K = 50$ and $L = 5$, Bernoulli-Gaussian model with $q = 0.5$:  Colors represent the relative gain of RCoF versus CDPC (e.g., ratio of sum rates $R_{\mbox{\tiny{RCoF}}}/R_{\mbox{\tiny{CDPC}}}$).}
\label{simulation_downlinkcol}
\end{figure}

We compared the proposed schemes with QF (uplink) and CDPC (downlink) over the Bernoulli-Gaussian model.
Recall that QF is a special case of QMF without binning, whose achievable sum rate is given in (\ref{eq:sumrateQMFs}).
In QF, more observations (i.e., more active ATs) generally improve the sum rate and thus AT selection is not needed
for the sake of maximizing the sum rate. Yet, for a fair comparison with the same total backhaul capacity, we considered
a greedy search that selects $L' = K < L$ active ATs, by maximizing at each step the
achievable sum rate. From Fig.~\ref{simulation_uplinkcol}, we observe that CoF outperforms QF when $R_{0}$ is small relatively to the
channel SNR. In this regime, the quantization noise dominates with respect to the non-integer penalty.
Instead, when $R_{0}$ increases, eventually QF outperforms CoF.
Fig.~\ref{simulation_downlinkcol} presents a comparison between RCoF and CDPC, leading to similar conclusions
for the DAS downlink.

Next, we examine the performance of the proposed low-complexity schemes QCoF, LQF, and RQCoF, by focusing on a small cell network scenario,
where ATs and UTs are close to each other. This yields  reflected by consider a fixed and relatively large SNR value ($\SNR = 25$ dB in our simulation), and comparing performances versus $R_0$, which becomes the main system bottleneck.
Fig.~\ref{simulation_LQFR0} shows that QCoF and LQF are competitive with respect to the performance of QF,
with significantly lower decoding  complexity. Furthermore, an additional remarkable feature of the lattice-based schemes
is that they can substantially reduce the channel state information overhead.
When QCoF (or LQF) with $p=7$ is used, each AT $\ell$ only requires $2K\log(7) \approx 28$ (with $K=5$)
bits of feedback per scheduling slot in order to forward the integer combination coefficients
(i.e., $\qv_{\ell}=(q_{\ell,1},...,q_{\ell,K})$) to the CP.

In Fig.~\ref{simulation_RQCoFR0}, RQCoF with $p=17$ can achieve the same spectral efficiency
of CDPC for $R_{0} \leq 5$ bits  and outperforms CZFB in the range of $R_{0} \leq 6$ bits.
For CZFB, we made use of the standard greedy user selection approach \cite{Dimic,Yoo} to find a subset
of $K' < K$ (with $K'=5$ in our simulation) active UTs. As expected, QCoF (resp., RQCoF) can achieve the performance of
QF (resp., CDPC) when the wired backhaul rate $R_0$ is not over-dimensioned with respect to the capacity
of the wireless channel. These observations point out that the proposed schemes are suitable for
low-complexity implementation of cooperative home networks where small home-based
access points are connected to the CP via digital subscriber line (DSL).

\begin{figure}
\centerline{\includegraphics[width=11cm]{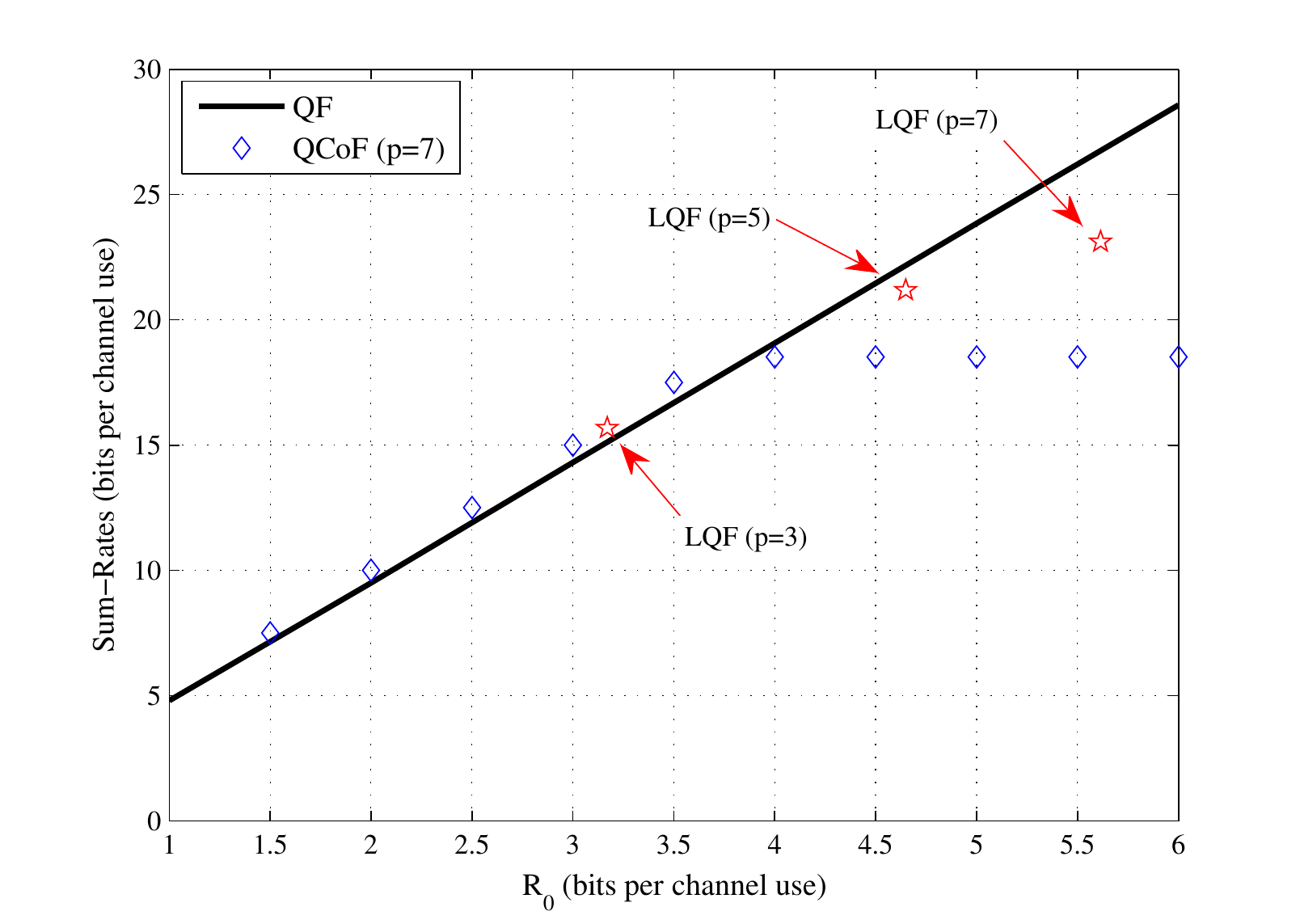}}
\caption{DAS uplink with $\SNR=25$ dB, $K = 5$ and $L = 50$: achievable sum rates as a function of $R_{0}$.}
\label{simulation_LQFR0}
\end{figure}

\begin{figure}
\centerline{\includegraphics[width=11cm]{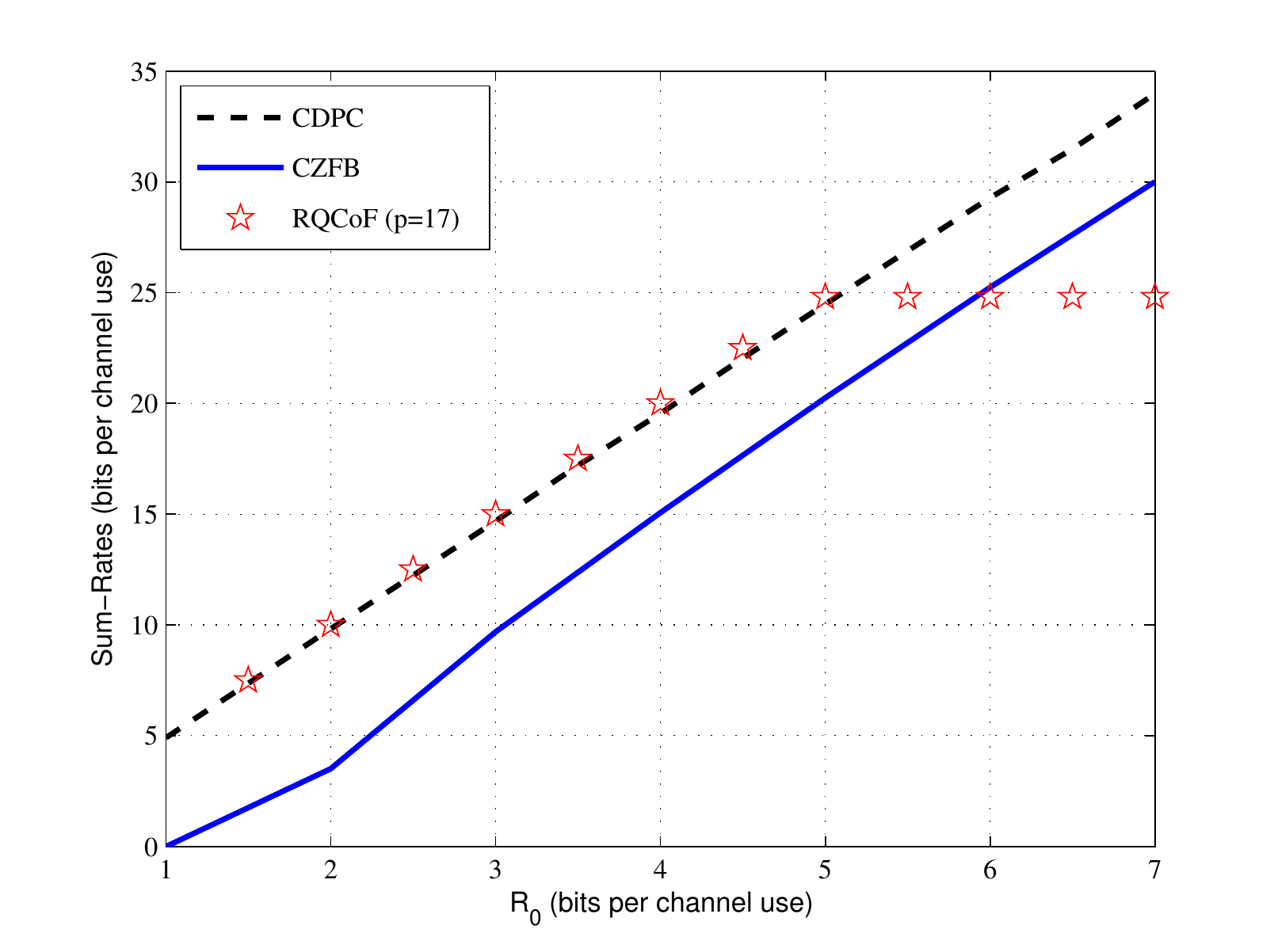}}
\caption{DAS downlink with $\SNR=25$ dB, $K = 50$ and $L = 5$: achievable sum rates as a function of $R_{0}$.}
\label{simulation_RQCoFR0}
\end{figure}

\section{Integer-Forcing Beamforming for the High-Capacity Backhaul Case} \label{sec:IFB}

CDPC (downlink) and QMF (uplink) are known to be optimal in the limit of $R_{0} \rightarrow \infty$. In fact, they converge to
the capacity achieving schemes of the corresponding MIMO G-BC and G-MAC channel models.  In this case, CoF and RCoF have no merit
because the impact of the non-integer penalty does not vanish as $R_0$ increases.
In this section we focus on the downlink in the regime of $R_0 \rightarrow \infty$. DPC is notoriously difficult to be implemented in practice,
since it requires nested lattice coding with shaping lattice $\Lambda$ of high dimension (see for example \cite{Erez2005,Bennatan2006}).
On the other hand, it is well-known that restricting the shaping lattice to have low dimension, in order to make the modulo-$\Lambda$ operation of manageable complexity,
does not provide significant performance benefits with respect to the simple Tomlinson-Harashima precoding approach,
which is equivalent to perform shaping with the cubic  lattice $\Lambda = \tau \ZZ[j]$ \cite{Windpassinger,Boccardi,CaireDIMACS}. Furthermore, it is also known that Tomlinson-Harashima
precoding for the MIMO G-BC does not provide significant gains with respect to simpler linear beamforming techniques, especially when
user selection and multiuser diversity can be exploited \cite{Yoo}.
Therefore, linear beamforming schemes are often proposed as a viable tradeoff between performance and complexity. When multiuser diversity cannot be exploited (e.g.,
the number of UTs $K$ is not large), linear beamforming may suffer from significant performance degradation when the channel matrix is
near singular.

For the uplink case, \cite{Jiening} show proposes an {\em Integer-Forcing Receiver} (IFR) that can approach the
performance of joint decoding with lower complexity and significantly outperforms the traditional linear multiuser detector schemes (e.g., the decorrelator or the linear MMSE detector) concatenated with single-user decoding.
The main idea is that the receiver antennas are used to create an effective channel matrix with integer-valued coefficients,
and CoF is used for the resulting integer-valued channel matrix, incurring no non-integer penalty.

In this section, we present a new beamforming strategy called Integer-Forcing Beamforming (IFB), that produces a similar effect
for the downlink. The precoding matrix $\Bm=[\bv_{1},\ldots,\bv_{L}]^{\transp}$ is chosen such that
the resulting effective channel matrix $\tilde{\Hm} \Bm$ is integer valued, i.e., $\tilde{\Hm}\Bm = \tilde{\Am}$,
with $\Bm=\tilde{\Hm}^{-1}\tilde{\Am}$ for some integer matrix $\tilde{\Am}$.
Then, RCoF can be applied as described in Section \ref{sec:DASRCoF}, to the resulting integer-valued effective channel matrix,
incurring no non-integer penalty. In short, IFB removes the non-integer penalty of RCoF but introduces a power
penalty (as in ZFB) due to the non-unitary precoding matrix $\Bm$. Notice that if we restrict $\tilde{\Am} = \Id$, then IFB coincides with ZFB.
Therefore, by allowing $\tilde{\Am}$ to be a general integer matrix, IFB performs at least as good as ZFB, and usually
significantly outperforms it, since its power penalty can be greatly reduced.  Although not investigated further in this work, we observe
here that a more general family of scheme might be devised by trading off the linear precoder  power penalty with the
RCoF non-integer penalty, by imposing an ``approximated'' integer forcing condition.

The detailed procedures of IFB  for a given $\tilde{\Am}$ (to be optimized later) is as follows:
\begin{itemize}
\item {\em Precoding over $\FF_{p^{2}}$ to eliminate integer-valued interferences}: Following the RCoF scheme
of Section \ref{sec:DASRCoF}), the CP precodes the zero-padded information messages $\{\underline{\tilde{\wv}}_\ell\}$ using
$\tilde{\Qm}^{-1}=g^{-1}([\tilde{\Am}] \mod p\ZZ[j])$ as in (\ref{eq:precod}), encodes the precoded messages $\{\underline{\tilde{\muv}}_\ell\}$
into the codewords $\{\underline{\tilde{\nuv}}_\ell\}$ and generates the  channel inputs
$\underline{\tilde{\xv}}_{\ell} = [\underline{\tilde{\nuv}}_\ell + \underline{\tilde{\dv}}_\ell] \mod \Lambda$, for $\ell=1,\ldots,L$, where
$\underline{\tilde{\dv}}_\ell$ are dithering sequences, as in (\ref{eq:channelinput}).
\item {\em Precoding over $\CC$ to create integer-valued channel matrix}: Using $\Bm = \tilde{\Hm}^{-1} \tilde{\Am}$, the CP produces
the precoded channel inputs
\begin{equation} \label{B-precoding}
\left [ \begin{array}{c}
\underline{\tilde{\vv}}_{1} \\ \vdots \\ \underline{\tilde{\vv}}_{L} \end{array} \right ] = \Bm
\left [ \begin{array}{c}  \underline{\tilde{\xv}}_{1} \\ \vdots \\ \underline{\tilde{\xv}}_{L} \end{array} \right ].
\end{equation}
\item Letting $k_\ell$ denote the index of the UT destination of the $\ell$-th message, its received signal is given by
\begin{equation}
\underline{\tilde{\yv}}_{k_\ell} = \sum_{\ell' =1}^L  \tilde{a}_{k_\ell,\ell'} \underline{\tilde{\xv}}_{\ell'} + \underline{\tilde{\zv}}_{k_\ell}.
\end{equation}
This is a G-MAC channel as in (\ref{GMAC}), with integer channel coefficients. Therefore,  IFB has eliminated
the non-integer penalty of RCoF. Finally, it is immediate to check (same steps as in (\ref{minchia})), that the
integer-valued interferences is eliminated by RCoF, i.e.,
each UT $k_\ell$ decodes its own lattice code $\Lc_\ell$ without multiuser interference.
\end{itemize}
The per-antenna power constraint imposes
\begin{equation}\label{eq:power}
\frac{1}{n} \EE \left [ \| \underline{\tilde{\vv}}_\ell\|^2 \right ] \leq \SNR, \;\; \mbox{for} \; \ell = 1,\ldots, L.
\end{equation}
From (\ref{B-precoding}), we have
\begin{eqnarray}
\frac{1}{n} \EE \left [ \| \underline{\tilde{\vv}}_\ell\|^2 \right ] & = &
\sum_{\ell' =1}^L  \frac{1}{n} \EE \left [ \| \underline{\tilde{\xv}}_{\ell'} \|^2 \right ] |b_{\ell',\ell}|^2.
\end{eqnarray}
Since $\underline{\tilde{\xv}}_\ell$ is uniformly distributed on $\Vc_\Lambda$, the constraint (\ref{eq:power}) yields
\begin{equation}
\frac{1}{n} \EE \left [ \| \underline{\tilde{\xv}}_{\ell} \|^2 \right ] = \sigma_\Lambda^2 = \frac{\SNR}{\max \{\|\bv_{\ell'}\|^2 : \ell' = 1,\ldots, L\}}, \;\; \mbox{for} \; \ell = 1,\ldots, L.
\end{equation}
Hence, we have:

\begin{theorem} IFB applied to a MIMO G-BC with channel matrix $\tilde{\Hm} \in \CC^{L \times L}$ achieves
the sum rate
\begin{equation}
R_{\mbox{\tiny{IFB}}}(\tilde{\Hm},\tilde{\Am}) = \sum_{\ell=1}^{L} R\left (\tilde{\av}_{\ell},\tilde{\av}_{\ell}, \SNR/\max\{\|\bv_{\ell'} : \ell'=1,\ldots, L\} \|^2 \right ),
\end{equation}
where we let $\tilde{\Am} = [\tilde{\av}_{1},\ldots,\tilde{\av}_{L}]^{\transp}$ and $\tilde{\Hm}^{-1}\tilde{\Am}
= [\bv_{1},\ldots,\bv_{L}]^{\transp}$.
\hfill \IEEEQED
\end{theorem}

The optimization of $\tilde{\Am}$ as a function of $\tilde{\Hm}$ appears to be a hard integer-programming problem without any particular structure lending itself
to computationally efficient algorithms. Instead, we resort to the suboptimal approach of optimizing $\tilde{\Am}$ with respect to the {\em sum power}, which is
proportional to  $\trace \left ( \Bm \Bm^\herm \right )$. Hence, the sum-power minimization problem takes on the form
\begin{eqnarray}
\min_{\tilde{\Am} \in \ZZ^{L\times L}[j]} &&  \trace \left ( \tilde{\Hm}^{-1} \tilde{\Am} \tilde{\Am}^\herm \tilde{\Hm}^{-\herm} \right ) \nonumber \\
\mbox{subject to} & & \mbox{Rank}(\tilde{\Am}) = L.\label{IFB:opt}
\end{eqnarray}
Writing $\trace \left ( \tilde{\Hm}^{-1} \tilde{\Am} \tilde{\Am}^\herm \tilde{\Hm}^{-\herm} \right ) = \sum_{\ell=1}^L \| \tilde{\Hm}^{-1} \tilde{\Am}([1:L],\ell) \|^2$
where $\tilde{\Am}([1:L],\ell) $ is
the $\ell$-th column of $\tilde{\Am}$, we notice that problem (\ref{IFB:opt}) is equivalent to finding
a reduced basis for the lattice generated by $\tilde{\Hm}^{-1}$.
In particular, the reduced basis takes on the form $\tilde{\Hm}^{-1} \Um$
where $\Um$ is a unimodular matrix over $\ZZ[j]$.  Hence, choosing $\tilde{\Am} = \Um$ yields the minimum sum-power subject
to the full rank condition in (\ref{IFB:opt}).  In practice, we used the (complex) LLL algorithm \cite{Napias}, with refinement of the LLL reduced
basis approximation by Phost or Schnorr-Euchner lattice search.

\begin{figure}
\centerline{\includegraphics[width=11cm]{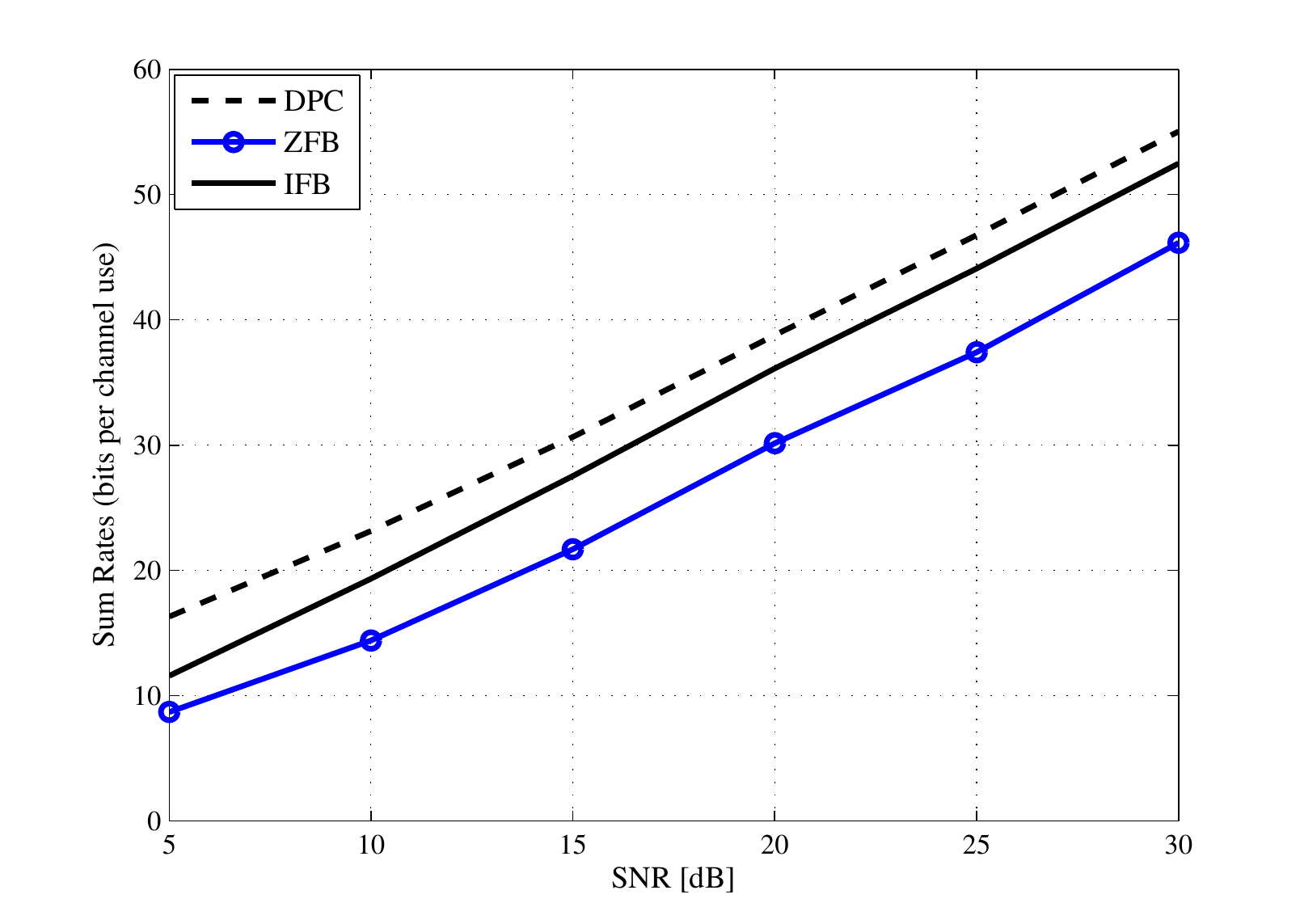}}
\caption{Achievable ergodic sum-rates as a function of SNRs for a MIMO-BC with same number $L = K = 5$ of ATs (transmit antennas) and UTs (users), over independent Rayleigh fading.}
\label{simulation_IFB}
\end{figure}

We consider the DAS downlink with infinite backhaul capacity with $5$ ATs and $5$ UTs.
The channel matrix $\tilde{\Hm}$ has i.i.d. elements $\tilde{h}(k,\ell) \sim \Cc\Nc(0,1)$ (independent Rayleigh fading).
Fig.~\ref{simulation_IFB} shows the ergodic achievable sum rate of IFB, compared with the ergodic channel sum capacity achieved by
DPC and by the sum rate achievable by ZFB. We notice that the proposed IFB downlink scheme
significantly improves over ZFB and approaches the sum capacity within $\approx 0.5$ bits per user.

\section{Conclusions}\label{sec:conclusion}

We considered a Distributed Antenna System (DAS) where several Antenna Terminals (ATs) are connected to a Central Processor (CP) via
digital backhaul links of rate $R_0$ bit/s/Hz. The ATs communicate with several User Terminals (UTs) simultaneously and on the same bandwidth,
such that the signals sent by the ATs interfere at each UT (downlink) and, Vice Versa, the signals sent by the UTs interfere at each AT (uplink).
The DAS uplink is a special case of a three-layers multi-source single-destination relay network, where the ATs play the role of the relays.
The DAS downlink is a special case of a relay broadcast network with one sender and individual messages.
For this setup, we considered the application of the Compute and Forward approach in various forms.
For the DAS uplink, CoF applies directly. In this case, we proposed system optimization based on network decomposition and
on greedy selection of the ATs for a given set of desired active UTs.
For the DAS downlink, we proposed a novel scheme referred to as Reverse CoF (RCoF). This scheme reverse the role of ATs and UTs with respect to
the uplink, and uses linear precoding over the finite field domain in order to eliminate multiuser interference.
In this case, we considered system optimization consisting of selecting a subset of UTs for a given set of active ATs.
It turns out that in this case the problem can be formulated as the maximization of a linear function subject to a matroid constraint, for which
a simple greedy procedure is known to be optimal. We also considered strategies that incorporate the presence of a ADC at the receiver as an unavoidable
part of the channel model. In this case, we can design lattice based strategies that explicitly take into account the presence of the finite resolution
scalar quantizer at the receivers. In particular, this leads to very simple single-user linear coding schemes
over $\FF_q$ with $q = p^2$, and $p$ a prime. Our own results in \cite{songnamITW} and others' results in \cite{Tunali,Feng} show that
it is possible to approach the theoretical performance of random coding using $q$-ary LDPC codes with
linear complexity in the code block length and polynomial complexity in the number of network nodes.
For the regime of large $R_0$, we have also introduced a novel linear precoding scheme referred to as Integer Forcing
Beamforming (IFB). This can be seen as a generalization of zero-forcing beamforming, where
the beam formed channel is forced to have integer coefficients, rather than to a diagonal matrix.
Then, RCoF can be applied to precode over the integer-valued multiuser downlink channel, without further non-integer penalty.

We provided extensive comparison of the proposed lattice-based strategies with
information-theoretic strategies for the DAS uplink and downlink, namely QMF and CDPC, known to be near-optimal.
We observed that the proposed strategies achieve similar and sometimes better performance in certain relevant regimes,
while providing a clear path for practical implementation, while the information-theoretic terms of comparisons are notoriously
difficult to be implemented in practice. As a matter of fact, today's technology relies on the widely suboptimal decode
and forward (DF) scheme for the uplink, or on the compressed linear beamforming approach for the downlink, which are
easily outperformed by the proposed schemes with similar, if not better, complexity.

As a conclusion, we wish to point out that the proposed schemes are competitive when the wired backhaul
rate $R_0$ is a limiting factor of the overall system sum rate.
For example, in a typical home  Wireless Local Area Network setting, the rates supported by the wireless segment are of the order of
10 to 50 Mbit/s, while typical DSL connection between the wireless router and the DSL central office (playing the role of the CP in our scenario)
has rates between 1 and 10 Mb/s. In this case, the schemes proposed in this paper can provide a viable and practical approach to uplink and
downlink centralized processing at manageable complexity.

\appendices

\section{Gaussian Approximation}\label{app:GA}

Let $\varepsilon = (p/\tau) \Re\{ \xi_{i}(\hv,\av,\alpha)\} \sim \Cc\Nc(0,\sigma_{\varepsilon}^2)$ with $\sigma^2_\varepsilon = \sigma_\xi^2/2$.
We consider the distribution of the discrete random variable $\nu = Q_{\ZZ}(\varepsilon)$. The pmf of
$\zeta_{i}(\hv,\av,\alpha)$ is obtained by considering i.i.d. real and imaginary parts, both distributed as $\nu$.  Define the function
\begin{eqnarray}
\Phi(x) &\triangleq& \PP\Big(\varepsilon > \frac{(2x-1)}{2}\Big) - \PP\Big(\varepsilon > \frac{(2x+1)}{2}\Big)\label{eq:Qfun}\\
&=& Q\Big(\frac{(2x-1)}{2\sigma_{\varepsilon}}\Big)-Q\Big(\frac{(2x+1)}{2\sigma_{\varepsilon}}\Big)\nonumber
\end{eqnarray}
where $Q(z) = \frac{1}{\sqrt{2\pi}}\int_{z}^{\infty}\exp\Big(-\frac{t^2}{2}\Big)dt$ is the Gaussian tail function. Recall that $g$ maps the $\ZZ_{p} =\{0,1,...,p-1\}$ into the set of integers $\{0,1,...,p-1\} \subset \RR$. We define an interval $\Ic(x)$ by
\begin{equation}
\Ic(x) \triangleq [x-0.5,x+0.5].
\end{equation} The pmf of $\nu$ can be computed as
\begin{equation}
\PP(\nu = \beta) \triangleq \PP\left(\varepsilon \in \bigcup_{m \in \ZZ} \Ic(g(\beta)+pm)\right).
\end{equation} For any $\beta_{1}, \beta_{2} \neq 0$ satisfying $g(\beta_{1}) + g(\beta_{2}) = p$, we have $\PP(\nu=\beta_{1}) = \PP(\nu=\beta_{2})$, which can be immediately proved using the symmetry of Gaussian distribution (about origin):
\begin{eqnarray}
&&\PP\left(\varepsilon \in \bigcup_{m \in \ZZ} \Ic(g(\beta_{1})+pm)\right)\\
&=&\PP\left(\varepsilon \in \bigcup_{m \in \ZZ_{+}\cup\{0\}} \Ic(g(\beta_{1})+pm)\right) + \PP\left(\varepsilon \in \bigcup_{m \in \ZZ_{-}\cup\{0\}} \Ic(g(\beta_{1})-p+pm)\right)\\
&=&\PP\left(\varepsilon \in \bigcup_{m \in \ZZ_{+}\cup\{0\}} \Ic(g(\beta_{1})+pm)\right) + \PP\left(\varepsilon \in \bigcup_{m \in \ZZ_{+}\cup\{0\}} \Ic(p-g(\beta_{1})+pm)\right)\label{eq:prob}\\
&=&\PP\left(\varepsilon \in \bigcup_{m \in \ZZ_{+}\cup\{0\}} \Ic(g(\beta_{1})+pm)\right) + \PP\left(\varepsilon \in \bigcup_{m \in \ZZ_{+}\cup\{0\}} \Ic(g(\beta_{2})+pm)\right)
\end{eqnarray} where $\ZZ_{+}$ and $\ZZ_{-}$ denote the positive and negative integers, respectively. Thus, we only need to find the pmf of $\nu$ with $\nu \leq \frac{p-1}{2}$ and other probabilities are directly obtained by symmetry. Using the (\ref{eq:prob}) for $\beta \neq 0$, we can quickly compute the pmf of $\nu$ using the $\Phi(x)$ defined in (\ref{eq:Qfun}):
\begin{eqnarray}
\PP(\nu = 0) &=& \Phi(0) + 2\sum_{m \in \ZZ_{+}} \Phi(g(\beta)+pm)\label{eq:prob1}\\
\PP(\nu = \beta) &=& \sum_{m \in \ZZ_{+} \cup \{0\}} \Phi(g(\beta)+pm)+  \Phi(p-g(\beta)+pm).\label{eq:prob2}
\end{eqnarray} In fact, $\Phi(x)$ is monotonically decreasing function on $x$ and in general, quickly converges to $0$ as $x$ increases. Therefore, we only need a finite number of summations in (\ref{eq:prob1}) and (\ref{eq:prob2}) and we observed that it is enough to sum
over $m=0,1,2$ in all numerical results presented in this paper.

\section{Proof of Theorem \ref{lem:LQF}}\label{proof:FF-MAC}

Consider the FF-MAC defined by $\yv = \Qm\xv \oplus \zetav$ where $\xv=(x_{1},...,x_{K})^{\transp} \in \FF_{p^2}^{K}$ and $\yv=(y_{1},...,y_{K})^{\transp} \in \FF_{p^2}^{K}$. The capacity region is the union of the rate regions defined by \cite{Kim}
\begin{equation} \label{region}
\sum_{k \in \Sc} R_k \leq I\left (\{x_k : k \in \Sc); \yv | \{x_k : \in \Sc^c\}, q \right ), \;\;\; \forall \;\; \Sc \subseteq [1:K],
\end{equation}
over all pmfs $P_{\xv, q} = P_q \prod_{k=1}^K P_{x_k |q}$. Since for any fixed such pmf the region (\ref{region}) is a polymatroid, the maximum
sum rate achieved on the dominant face $\sum_{k=1}^K R_k = I(\xv; \yv | q)$.  Since the expectation of the maxima is larger or equal to the maximum
of the expectation,  we have that $\sum_{k = 1}^L  R_{k} \leq \max_{P_{\xv,q}} I(\xv;\yv|q)$.
Finally, since $q \rightarrow \xv \rightarrow \yv$, we have:
\begin{equation}
I(\xv;\yv|q) \leq I(\xv, q; \yv) = I(\xv; \yv) + I(q;\yv|\xv) = I(\xv;\yv),
\end{equation}
showing that time-sharing is not needed for the maximum sum rate. Since $\Qm$ is full rank, uniform i.i.d. inputs $\FF_{p^2}$ achieve
\begin{eqnarray}
I(\xv;\yv) &=& H(\yv) - H(\yv|\xv) = H(\yv) - H(\zetav)\\
&\leq& 2K\log{p} - H(\zetav).
\end{eqnarray}
Finally, since $\sum_{k=1}^{K} H(\zeta_{k}) \geq H(\zeta_{1},...,\zeta_{K})$,
we conclude that the sum rate in (\ref{sucacazzi}) is achievable.

\section{Proof of Theorem \ref{lem:sumRQCoF}}\label{proof:FF-BC}

We consider the FF-BC defined by $\yv = \tilde{\Qm}\xv \oplus \zetav$, where $\xv=(x_{1},...,x_{L})^{\transp}$ and $\yv=(y_{1},...,y_{L})^{\transp}$.
Since $\Qm$ is invertible, letting $\xv = \tilde{\Qm}^{-1}\vv$ for $\vv \in \FF_{p^2}^L$ yields the orthogonal BC
$\yv = \vv \oplus \zetav$. The achievable sum rate for this decoupled channel is obviously given by the sum of of the capacities
of each individual additive-noise  finite-field channel, irrespectively of the statistical dependence across the noise components. Each $\ell$-th channel
capacity is achieved by letting $\vv$ i.i.d. with uniformly distributed components over $\FF_{p^2}$. It follows that the sum rate
(\ref{straminchia}) is achievable. In order to show that this is in fact the sum-capacity of the FF-BC, we notice that
a trivial upper-bound on the broadcast capacity region is given by \cite{Kim}:
\begin{equation}
R_{\ell} \leq \max_{P_{\xv}} \; I(\xv;y_{\ell}) \;\; \mbox{for} \;\; \ell=1,\ldots,L.
\end{equation}
This is the capacity of the single-user channel with transition probability $P_{y_{\ell}|\xv}$.
Due to the additive noise nature of the channel, we have $I(\xv;y_{\ell}) = H(y_{\ell}) - H(\zeta_{\ell})$.
Furthermore, $H(y_{\ell}) \leq 2 \log{p}$ and this upper bound is achieved by letting $\xv \sim $Uniform over $\FF_{p^2}^L$.
Summing over $\ell$ we find that the upper bound on the sum capacity coincides with (\ref{straminchia}).

\section*{Acknowledgment}

The authors would like to thank Bobak Nazer and Or Ordentlich for sharing their encouraging feedback and
insightful observations.


\clearpage

\end{document}